\documentclass[12pt,letterpaper]{article}

\usepackage{amsmath,amssymb,amsthm,amsfonts}
\allowdisplaybreaks[2]
\usepackage[margin=1in]{geometry}
\usepackage{indentfirst}

\usepackage[T1]{fontenc}
\IfFileExists{newtxtext.sty}{\usepackage{newtxtext,newtxmath}}{\usepackage{mathptmx}}
\usepackage[moderate]{savetrees}
\makeatletter
\renewcommand\section{\@startsection{section}{1}{\z@}{-3.5ex \@plus -1ex \@minus -.2ex}{2.3ex \@plus .2ex}{\normalfont\Large\bfseries}}
\renewcommand\subsection{\@startsection{subsection}{2}{\z@}{-3.25ex \@plus -1ex \@minus -.2ex}{1.5ex \@plus .2ex}{\normalfont\large\bfseries}}
\renewcommand\subsubsection{\@startsection{subsubsection}{3}{\z@}{-3.25ex \@plus -1ex \@minus -.2ex}{1.5ex \@plus .2ex}{\normalfont\normalsize\bfseries}}
\renewcommand\paragraph{\@startsection{paragraph}{4}{\z@}{.8ex plus .2ex minus .1ex}{-.5em}{\normalfont\normalsize\bfseries}}
\makeatother

\usepackage{graphicx}
\usepackage{float}
\usepackage{xcolor}
\usepackage{tikz}
\usetikzlibrary{arrows.meta}

\definecolor{referenceblue}{RGB}{0,76,153}
\usepackage{enumerate}
\usepackage{array}
\usepackage{etoolbox}
\usepackage{needspace}
\AtBeginEnvironment{thebibliography}{\interlinepenalty=10000}

\makeatletter
\def\@listi{\leftmargin\leftmargini
  \topsep 3pt plus 1pt minus 1pt
  \partopsep 0pt
  \parsep 0pt
  \itemsep 2pt plus .5pt}
\let\@listI\@listi
\def\@listii{\leftmargin\leftmarginii
  \labelwidth\leftmarginii
  \advance\labelwidth-\labelsep
  \topsep 2pt plus 1pt minus 1pt
  \parsep 0pt
  \itemsep 1pt plus .5pt}
\makeatother
\newcommand{\afterfigurespace}{\par\vspace{2pt}}
\newsavebox{\fitwidthbox}
\newcommand{\fitwidth}[1]{\sbox{\fitwidthbox}{#1}\ifdim\wd\fitwidthbox>\linewidth
    \resizebox{\linewidth}{!}{\usebox{\fitwidthbox}}\else
    \usebox{\fitwidthbox}\fi}
\usepackage[round]{natbib}

\usepackage{setspace}
\AtBeginEnvironment{figure}{\normalsize\setstretch{1.5}}

\usepackage[hang,flushmargin]{footmisc}

\usepackage[hypertexnames=false,hyperfootnotes=false]{hyperref}
\hypersetup{
  colorlinks=true,
  linkcolor=referenceblue,
  citecolor=referenceblue,
  urlcolor=referenceblue,
  filecolor=referenceblue,
  pdftitle={Equilibrium Architecture in Multi-Battle Contests with Count-Dependent Prizes},
  pdfauthor={Zhonghong Kuang and Jingfeng Lu},
  pdfsubject={Equilibrium in multi-battle Tullock contests under general count-dependent prize schedules},
  pdfkeywords={multi-battle contests, Tullock contests, count-dependent prizes, mixed-strategy equilibrium, endogenous participation}
}
\AtBeginDocument{\clubpenalty=10000\widowpenalty=10000\displaywidowpenalty=10000}

\newtheoremstyle{compactplain}
  {2pt}
  {2pt}
  {\itshape}
  {}
  {\bfseries}
  {.}
  {.5em}
  {}

\newtheoremstyle{compactdefinition}
  {2pt}
  {2pt}
  {\normalfont}
  {}
  {\bfseries}
  {.}
  {.5em}
  {}

\theoremstyle{compactplain}
\newtheorem{theorem}{Theorem}
\newtheorem{proposition}{Proposition}
\newtheorem{claim}{Claim}
\newtheorem{corollary}{Corollary}
\newtheorem{lemma}{Lemma}

\theoremstyle{compactdefinition}
\newtheorem{definition}{Definition}
\newtheorem{example}{Example}

\newcommand{\numberappendixobjectsbysection}{\numberwithin{theorem}{section}\numberwithin{proposition}{section}\numberwithin{claim}{section}\numberwithin{corollary}{section}\numberwithin{lemma}{section}\numberwithin{assumption}{section}\numberwithin{definition}{section}\numberwithin{example}{section}\numberwithin{remark}{section}\numberwithin{equation}{section}\numberwithin{figure}{section}\numberwithin{table}{section}}

\newcommand{\compactdisplays}{\setlength{\abovedisplayskip}{3pt plus 1pt minus 1pt}\setlength{\belowdisplayskip}{3pt plus 1pt minus 1pt}\setlength{\abovedisplayshortskip}{2pt plus 1pt minus 1pt}\setlength{\belowdisplayshortskip}{2pt plus 1pt minus 1pt}}
\apptocmd{\normalsize}{\compactdisplays}{}{}
\AtBeginDocument{\compactdisplays}

\title{Equilibrium Architecture in Multi-Battle\\Contests with Count-Dependent Prizes\thanks{We thank Qian Jiao and Yuxuan Zhu for helpful comments; all remaining errors are our own. Kuang gratefully acknowledges financial support from the National Natural Science Foundation of China (NSFC), grant numbers~72103190 and 72473137. Kuang: School of Economics, Renmin University of China, 59 Zhongguancun Street, Beijing 100872, China; email: \href{mailto:kuang@ruc.edu.cn}{\texttt{kuang@ruc.edu.cn}}. Lu: Department of Economics, National University of Singapore, 1 Arts Link, Singapore 117570; email: \href{mailto:ecsljf@nus.edu.sg}{\texttt{ecsljf@nus.edu.sg}}.}}
\author{{\Large\textsc{Zhonghong Kuang}}
  \\[-.05em] {\large\textit{Renmin University of China}}
  \and
  {\Large\textsc{Jingfeng Lu}}
  \\[-.05em] {\large\textit{National University of Singapore}}}
\date{October 2026}

\makeatletter
\renewcommand{\maketitle}{{\centering
   {\LARGE\bfseries\@title\par}
   \vspace{1.15\baselineskip}
   {\normalsize
    \def\and{\end{tabular}\hspace{3em}\begin{tabular}[t]{c}}\begin{tabular}[t]{c}\@author\end{tabular}\par}
   \vspace{.9\baselineskip}
   {\large\@date\par}
   \vspace{.35\baselineskip}}\@thanks
}
\makeatother

\begin{document}

\renewcommand{\thefootnote}{\fnsymbol{footnote}}
\maketitle
\setcounter{footnote}{0}
\renewcommand{\thefootnote}{\arabic{footnote}}

\begin{abstract}
\noindent We study how aggregating local victories shapes participation and equilibrium uniqueness. Two contestants with possibly different marginal costs compete across identical Tullock battlefields with discriminatory power at most one. A nonconstant, nondecreasing, symmetric schedule divides a fixed prize by victory count. We establish existence and show that every realized equilibrium allocation divides effort equally across battlefields. Equilibrium is unique through six battlefields, with at most two effort levels per contestant; seven first permits multiple equilibria or randomization by both contestants. The prize schedule and technology, not the cost ratio, determine whether equilibria use infinitely many effort levels; in that case, cost asymmetry changes only the higher-cost contestant's participation probability. Under majority rule, we characterize the unique equilibrium for every battlefield count and cost ratio.


\end{abstract}

\clearpage

\section{Introduction}\label{S1}

Many competitions are fought on several fronts but settled by a rule that aggregates local victories: political parties across electoral districts, firms across research projects, and sports teams across matches. Majority rewards can sharpen incentives but make modest efforts unrewarding and induce withdrawal; smoother rewards preserve value for partial success but may weaken marginal incentives. These competing effects raise two questions: \emph{When does aggregation induce inactivity or strategic randomization? How complex can equilibrium become as the number of battlefields grows?}

We study two risk-neutral contestants competing simultaneously across \(J\) identical Tullock battlefields. On each battlefield, for nonnegative efforts \(x,y\), the first contestant wins with probability \(x^r/(x^r+y^r)\) when \(x+y>0\), and \(1/2\) at mutual inactivity; \(r\in(0,1]\), called discriminatory power, measures responsiveness to relative effort. Marginal effort costs are constant but may differ: we normalize the higher cost to one and denote the other by \(c\in(0,1]\); \(c=1\) means equal costs. A common schedule assigns prize share \(q_n\) for \(n\) victories. We call it admissible if it is nonconstant, monotone, and symmetric: \(q_n+q_{J-n}=1\). It thus treats contestants identically and divides one fixed prize. Examples include proportional rewards, \(q_n=n/J\), and majority rule with ties split equally.

To characterize equilibrium under these prize schedules, we must account for both effort allocation and participation. Identical battlefields alone do not justify uniform effort: contestants may allocate effort unevenly or randomize over allocations with correlated battlefield efforts. Inactivity also creates a payoff discontinuity: each contestant receives half the prize at mutual inactivity, but arbitrarily small positive uniform effort against an inactive rival yields \(q_J>1/2\). Moreover, an admissible schedule can make aggregate reward nonconcave in relative effort, so local first-order conditions alone neither identify mixed equilibrium nor establish uniqueness.

Despite these difficulties, we establish equilibrium existence and show that every realized equilibrium allocation is uniform: effort is equal across battlefields. Although contestants are free to allocate effort unevenly across battlefields, the uniformity result lets us analyze every equilibrium using each contestant's total effort alone, while still allowing inactivity and randomization. Equilibria may be pure, with both contestants choosing fixed total efforts; semi-pure, with only one randomizing over finitely many effort levels; or two-sided mixed, with both randomizing.

Our first main result identifies a sharp frontier for equilibrium uniqueness. For every admissible schedule, cost ratio, and \(r\in(0,1]\), equilibrium is unique with six or fewer battlefields (Theorem~\ref{thm:lowcount-refinements}). It is either pure or semi-pure; in the latter case, exactly one contestant mixes between two effort levels. Only the higher-cost contestant can mix through five battlefields; at six, the lower-cost contestant can also be the sole mixer. Seven first permits either a continuum of equilibria or mixing by both contestants. We give exact constructions establishing this frontier. Multiplicity can alter participation and effort dispersion without changing expected expenditure, prize shares, or payoffs.

Our second main result clarifies the role of cost asymmetry when contestants randomize their efforts. Only the strictly higher-cost contestant can be inactive with positive probability. Under two-sided mixing, the sets of possible positive efforts are either both finite or both countably infinite, accumulating only at zero in the latter case. This finite/infinite distinction is independent of the cost ratio (Theorem~\ref{thm:infinite-phase}): for a fixed admissible schedule and technology, an infinite-support equilibrium at one cost ratio implies infinite positive supports in every equilibrium at every cost ratio. In this phase, the full asymmetric-cost equilibrium set follows from the equal-cost set by changing only the higher-cost contestant's participation probability to \(c\).

Our third main result proves equilibrium uniqueness under majority rule for every battlefield count and cost ratio (Theorem~\ref{thm:majority-regimes}). Equilibrium is pure at low discriminatory power (high noise); at intermediate discriminatory power the cost ratio determines whether it is pure or semi-pure; and at high discriminatory power both contestants mix over countably infinite positive supports accumulating only at zero. In this last region, they share a unique effort distribution conditional on participation. We give explicit pure and semi-pure strategies, and establish uniqueness and support structure in the two-sided mixed region. The two-sided mixed majority region first appears at seven battlefields. At \(c=r=1\), the pure candidate implied by first-order conditions costs each contestant \(15/32\) of the prize with six battlefields but \(35/64\) with seven. The latter exceeds the expected prize share of \(1/2\), making withdrawal profitable. Figure~\ref{fig:majority-regions} summarizes the majority regimes.

\begin{figure}[H]
\centering
\scalebox{.95}{\begin{tikzpicture}[
  x=10.7cm,y=2.65cm,
  axis/.style={-{Latex[length=1.7mm]},thin},
  boundary/.style={very thick},
  guide/.style={densely dotted,thin},
  every node/.style={font=\small}
]
\def\rlow{0.24}
\def\rhigh{0.72}

\path[fill=gray!55] (\rhigh,0) rectangle (1,1);
\path[fill=gray!20]
  (\rlow,0)
  .. controls (0.39,0.03) and (0.61,0.58) .. (\rhigh,1)
  -- (\rhigh,0) -- cycle;
\draw[boundary]
  (\rlow,0)
  .. controls (0.39,0.03) and (0.61,0.58) .. (\rhigh,1);
\draw[guide] (\rlow,0) -- (\rlow,1);
\draw[guide] (\rhigh,0) -- (\rhigh,1);

\draw[axis] (0,0) -- (1.05,0)
  node[right,align=left] {discriminatory\\[-1pt]power $r$};
\draw[axis] (0,0) -- (0,1.07);
\node[rotate=90] at (-0.105,0.53) {cost ratio $c$};
\node[left,align=right] at (-0.012,1) {$c=1$};
\node[left,align=right] at (-0.012,0) {$c\to0$};
\node[below] at (0.12,-0.02) {low};
\node[below] at (0.49,-0.02) {intermediate};
\node[below] at (0.86,-0.02) {high};

\node[font=\bfseries\small,align=center] at (0.12,0.52) {Pure};
\node[font=\bfseries\small,align=center] at (0.51,0.78) {Pure};
\node[font=\bfseries\small,align=center] at (0.52,0.14) {Semi-pure};
\node[font=\bfseries\small,align=center] at (0.86,0.52)
  {Two-sided\\[-1pt]mixed};
\end{tikzpicture}}
{\normalsize\begin{spacing}{1}
\caption{Schematic majority-rule regions for \(J\geq7\); boundaries depend on \(J\).}
\label{fig:majority-regions}
\end{spacing}}
\end{figure}

\medskip
\noindent\textbf{Related literature.} Our analysis builds on the identity-independent model of \citet*[hereafter FKL]{fan2026generalized}. FKL show that, if a uniform equilibrium exists, every equilibrium is uniform. They also establish the sharp bound \(r\leq1/\lceil J/2\rceil\) guaranteeing a unique equilibrium that is pure for every admissible schedule and cost ratio, and characterize effort-maximizing prize rules under that condition. We establish existence and uniformity for all \(0<r\leq1\). Theorem~\ref{thm:lowcount-refinements} establishes uniqueness through six battlefields even when equilibrium is mixed.

The resulting scalar representation covers every equilibrium of the unrestricted vector game and connects our analysis to standard and asymmetric Tullock contests \citep{tullock1980efficient,nti1999rent,feng2017uniqueness} and complete-information all-pay contests \citep{baye1996allpay,siegel2009allpay,siegel2010conditional}. Other work studies equilibrium comparisons and participation in broader scalar contests \citep{malueg2005equilibria,alcalde2007tullock,levine2022success}. \Citet{ewerhart2015mixed} establishes countably infinite discrete supports in symmetric mixed equilibria of the standard two-player Tullock contest; \citet{ewerhart2025uniqueness} completes its uniqueness analysis. We extend the analytic support argument to general victory-count rewards and unequal costs. Theorem~\ref{thm:infinite-phase} shows that infinite-support mixing is independent of the cost ratio. For general victory-count rewards, however, uniqueness can fail; majority rule retains it.

For majority rule, the closest mixed-equilibrium benchmark is \citet[hereafter KP]{klumpp2006primaries}. For equal-cost contests with an odd battlefield count, KP establish uniform effort allocation, payoff equivalence, and a pure-to-mixed transition. Building on this benchmark, Theorem~\ref{thm:majority-regimes} proves uniqueness even when equilibrium is mixed. The theorem also characterizes the support structure and covers unequal costs and even battlefield counts. Section~\ref{S5:implications} provides a detailed comparison with KP. \Citet{luwangzhou2024electoral} derive equilibria for three-battlefield models with $r=1$, special cases of our setup.

More broadly, related work studies contests with multiple fronts. Simultaneous-object auctions generate strategic dependence across prizes \citep{szentes2003beyond,szentes2003three,ewerhart2022fractal}. Majority contests link battles through an aggregate threshold \citep{barelli2014majority}, whereas Colonel Blotto and limited-resource contests impose a common resource constraint across fronts \citep{roberson2006colonel,kvasov2007limited}. Other work examines multi-battle and team contests \citep{konrad2009multibattle,fu2015team} and prize design or bundling across dimensions \citep{feng2018split,lu2024performance,feng2024optimal}. Here, expenditure is endogenous and a fixed prize is divided by victory count. The schedule therefore links battlefield incentives and shapes participation and mixing.

\medskip
\noindent\textbf{Roadmap.} Section~\ref{S2} presents the model and scalar reduction. Section~\ref{S3} develops tools for analyzing equilibrium outcomes and strategies. Section~\ref{S4} studies uniqueness and multiplicity across battlefield counts. Section~\ref{S5} establishes the cost-independent infinite-support phase and uniqueness under majority rule. Section~\ref{S6} discusses extensions and concludes.

\section{Model and Equilibrium Reduction}\label{S2}

\subsection{Prize Schedule, Effort, and Normalization}\label{S2:prizes}

Two risk-neutral contestants, \(A\) and \(B\), compete across \(J\geq1\) identical battlefields \(j\in\mathcal J\equiv\{1,\ldots,J\}\). A \emph{count-dependent prize schedule} \(\boldsymbol q\equiv(q_0,\ldots,q_J)\in[0,1]^{J+1}\) assigns share \(q_n\) for \(n\) victories. The same schedule applies to both: if \(A\) wins \(n\), \(B\) wins \(J-n\) and receives \(q_{J-n}\). Under risk neutrality, \(q_n\) can also represent the probability of receiving one indivisible prize.

\begin{definition}[Admissible prize schedule]\label{def:admissible}
A prize schedule \(\boldsymbol q\) is \emph{admissible} if it satisfies:
\begin{enumerate}[(i)]
\item \emph{budget balance and contestant symmetry:} \(q_n+q_{J-n}=1\) for every \(n\in\{0,\ldots,J\}\);
\item \emph{monotonicity:} \(q_{n+1}\geq q_n\) for every \(n\in\{0,\ldots,J-1\}\); and
\item \emph{nonconstancy:} \(q_J>q_0\).
\end{enumerate}
\end{definition}

We maintain admissibility throughout. It implies \(0\leq q_0<1/2<q_J\leq1\) and, for even \(J\), \(q_{J/2}=1/2\), so an equal battlefield split divides the prize equally. Examples include proportional rewards, \(q_n=n/J\), and majority rule with ties split equally.

Subtracting \(q_0\) preserves incentives. Let \(\Lambda\equiv q_J-q_0\in(0,1]\) be the prize range, define \(\widetilde q_n\equiv(q_n-q_0)/\Lambda\), and write \(\widetilde{\boldsymbol q}\equiv(\widetilde q_0,\ldots,\widetilde q_J)\). Then \(\widetilde q_0=0\), \(\widetilde q_J=1\), and \(\widetilde q_n+\widetilde q_{J-n}=1\). The normalized schedule determines how rewards vary with victories, while \(\Lambda\) sets their scale. With linear costs and \(\widetilde{\boldsymbol q}\) fixed, replacing \(\Lambda\) by \(\Lambda'\in(0,1]\) scales every equilibrium effort by \(\Lambda'/\Lambda\), preserving effort ratios, participation, and mixing probabilities.

Contestants move simultaneously; we study Nash equilibrium. Each \(i\in\{A,B\}\) chooses a nonnegative effort vector \(\mathbf e_i=(e_{ij})_{j\in\mathcal J}\), with total effort \(E_i\equiv\sum_{j\in\mathcal J}e_{ij}\). Every battlefield has Tullock parameter \(r\in(0,1]\); \(A\)'s win probability on \(j\) is
\[
p_{j,A}=
\begin{cases}
\dfrac{e_{Aj}^{r}}{e_{Aj}^{r}+e_{Bj}^{r}},&e_{Aj}+e_{Bj}>0,\\[5pt]
\dfrac12,&e_{Aj}=e_{Bj}=0,
\end{cases}
\qquad
p_{j,B}=1-p_{j,A}.
\]
A contestant wins for sure if only she exerts positive effort. Higher \(r\) increases responsiveness to relative effort; \(r=1\) is the standard Tullock case. Against fixed positive rival effort, \(r\leq1\) makes the win probability concave in own effort. This concavity underlies the equalization argument below, although the aggregate reward need not be concave in effort. For \(r>1\), local convexity near zero can favor concentration. Battlefield outcomes are independent conditional on the effort profile. Let \(W_i\) be contestant \(i\)'s number of victories.

Costs are linear. Normalize \(B\)'s marginal cost to one and relabel so that \(A\)'s is \(c\in(0,1]\). Expected payoffs are \(\pi_A=\mathbb E[q_{W_A}]-cE_A\) and \(\pi_B=\mathbb E[q_{W_B}]-E_B\). Normalize these as \(\widehat\pi_i\equiv(\pi_i-q_0)/\Lambda\). This positive affine transformation preserves best responses and equilibria.

A contestant is \emph{inactive} at the zero vector and \emph{participates} otherwise. A pure strategy is an effort vector in \(\mathbb R_{\geq0}^J\); a mixed strategy is a distribution over these vectors with finite expected total effort. Inactivity, unequal battlefield efforts, and arbitrary correlation within a contestant's effort vector are allowed.

\subsection{From Battlefield Efforts to a Scalar Contest}\label{S2:reduction}

An effort vector is \emph{uniform} if its components are equal; an equilibrium is uniform if \(e_{ij}=E_i/J\) with probability one for every \(i,j\). In the \emph{scalar restriction}, each contestant chooses total effort and allocates it uniformly. Combining FKL's uniform-strategy reduction (Proposition~1) with existence, we show that this restriction preserves the equilibrium set.

\begin{lemma}[Equalization and interchangeability]\label{lem:core-reductions}
\begin{enumerate}[(i)]
\item Against any mixture over uniform effort vectors with positive participation probability, every best-response effort vector is uniform.
\item Any pair of equilibrium strategies, one per contestant, forms an equilibrium.
\item If one equilibrium is uniform and both contestants participate with positive probability, every equilibrium is uniform.
\end{enumerate}
\end{lemma}

At fixed total effort and cost, monotonicity and \(r\leq1\) make pairwise equalization weakly raise expected reward against a positive uniform rival. A single averaging step can be payoff-neutral when own efforts contain zeros. The proof first moves a nonuniform vector toward a positive nonuniform allocation of the same total without lowering reward. Equalization within the positive domain then strictly raises reward up to the uniform allocation, so complete equalization is strict even with initial zeros. Adding the opponent's cost to each payoff preserves best responses and makes the game constant-sum, giving interchangeability. Appendix~\ref{app:core-reductions} proves the lemma, including zero effort and mixed strategies.

Measure total effort in reward-range units: \(X_i\equiv E_i/\Lambda\). Uniform allocation assigns \(\Lambda X_i/J\) to each battlefield. For positive realized efforts, let \(t\) be own total effort divided by the rival's. Each battlefield is won with probability \(p=t^r/(1+t^r)\), so victories are binomial with parameters \(J,p\). Normalized expected reward is
\begin{equation}
Q_J(p)=\sum_{n=0}^{J}\widetilde q_n\binom{J}{n}p^n(1-p)^{J-n},\qquad
\varphi_{J,r}(t)=Q_J\!\left(\frac{t^r}{1+t^r}\right)
=\frac{\sum_{n=0}^{J}\binom{J}{n}\widetilde q_n t^{rn}}{(1+t^r)^J}.
\label{eq:normalized-payoff-kernel}
\end{equation}
Thus \(Q_J(p)\) and \(\varphi_{J,r}(t)\) express the same reward by win probability and effort ratio, respectively, with \(\varphi_{J,r}(t)=\varphi_{J,1}(t^r)\). Dependence on \(\widetilde{\boldsymbol q}\) is implicit; we suppress unambiguous subscripts. Prize symmetry gives \(\varphi(t)+\varphi(1/t)=1\) for \(t>0\), since inverting the ratio exchanges contestants. Set \(\varphi(0)=0\) and \(\varphi(\infty)=1\); these limits do not determine the payoff at mutual inactivity.

At normalized total efforts \(X_A=x\) and \(X_B=y\), \(A\)'s normalized prize is \(\mathcal U(x,y)=\varphi(x/y)\) for \(y>0\), with \(\mathcal U(x,0)=1\) for \(x>0\) and \(\mathcal U(0,0)=1/2\). Normalized payoffs are \(\mathcal U(x,y)-cx\) for \(A\) and \(1-\mathcal U(x,y)-y\) for \(B\). The following proposition holds for every \(J\geq1\), admissible schedule, \(c\in(0,1]\), and \(r\in(0,1]\).

\begin{proposition}[Existence and scalar reduction]\label{thm:existence}\label{prop:scalar-reduction}
An equilibrium exists. Every equilibrium is uniform, and the scalar restriction preserves the equilibrium set.
\end{proposition}

Appendix~\ref{app:existence-scalar-reduction} constructs a scalar equilibrium and proves positive participation by both contestants in every scalar equilibrium. A scalar equilibrium already rules out profitable uniform deviations. Lemma~\ref{lem:core-reductions}\emph{(i)} shows that a nonuniform deviation cannot do better, so each scalar equilibrium is also an equilibrium of the full contest. One uniform equilibrium and Lemma~\ref{lem:core-reductions}\emph{(iii)} then make every full-game equilibrium uniform.

Uniformity concerns each realized effort vector, not merely the expected allocation; total efforts may differ and remain random. Henceforth, strategies are distributions of normalized total effort.

\subsection{Equilibrium Forms and Participation}\label{S2:equilibrium-forms}

For a distribution \(F\) of normalized total effort, its \emph{support} \(\operatorname{supp}F\) consists of levels for which every open interval containing the level has positive probability; its \emph{positive support} is \(\operatorname{supp}F\cap(0,\infty)\). Support points need not themselves have positive probability. A \emph{point mass} concentrates probability one on a single effort; a \emph{nondegenerate} distribution does not. An \emph{atom at zero} assigns positive probability to inactivity; an \emph{accumulation point} is approached by infinitely many distinct support points. Equilibrium is \emph{pure} if both contestants fix total effort, \emph{semi-pure} if exactly one fixes positive total effort while the other mixes nontrivially, and \emph{two-sided mixed} if both mix nontrivially.
Let \(\delta_x\) denote unit mass at the scalar \(x\).

\begin{proposition}[Equilibrium taxonomy and support restrictions]\label{thm:equilibrium-taxonomy}
\leavevmode
\begin{enumerate}
\item[(i)] \textbf{\upshape Equilibrium and participation.} Every equilibrium is pure, semi-pure, or two-sided mixed. Contestant~\(A\) has no atom at zero, whereas contestant~\(B\) participates with positive probability and can have an atom at zero only if \(c<1\).
\item[(ii)] \textbf{\upshape Support restrictions.} A semi-pure mixer has finite support. In a two-sided mixed equilibrium, the positive supports are either both finite or both countably infinite, accumulating only at zero.
\end{enumerate}
\end{proposition}

Figure~\ref{fig:equilibrium-taxonomy} summarizes the five support patterns permitted by Proposition~\ref{thm:equilibrium-taxonomy}. Here \(\mathrm P\) denotes pure play; the subscripts in \(\mathrm{SP}_A\) and \(\mathrm{SP}_B\) identify the pure contestant. The symbols \(\mathrm M_f\) and \(\mathrm M_\infty\) denote two-sided mixing with finite or countably infinite positive supports, respectively. Filled dots denote positive support points; hollow dots in panels~(b), (d), and~(e) mark an optional atom at zero for \(B\), and ellipses indicate infinitely many positive support points accumulating at zero.

\begin{figure}[H]
\centering
\begin{tikzpicture}[
  axis/.style={-{Latex[length=1.2mm]},thin},
  atom/.style={circle,fill=black,inner sep=1.35pt},
  optional/.style={circle,draw=black,fill=white,inner sep=1.3pt},
  every node/.style={font=\footnotesize}
]
\foreach \dx in {0,3.05,6.10,9.15,12.20}{
  \begin{scope}[shift={(\dx,0)}]
    \draw[axis] (.22,.58) -- (2.62,.58);
    \draw[axis] (.22,0) -- (2.62,0);
    \node[left] at (.18,.58) {$A$};
    \node[left] at (.18,0) {$B$};
    \node[below] at (.22,-.03) {$0$};
  \end{scope}
}
\begin{scope}
  \node at (1.42,1.06) {\textbf{(a)} \(\mathrm P\)};
  \node[atom] at (1.50,.58) {};
  \node[atom] at (1.15,0) {};
\end{scope}
\begin{scope}[shift={(3.05,0)}]
  \node at (1.42,1.06) {\textbf{(b)} \(\mathrm{SP}_A\)};
  \node[atom] at (1.50,.58) {};
  \node[optional] at (.22,0) {};
  \node[atom] at (1.02,0) {};
  \node[atom] at (1.88,0) {};
\end{scope}
\begin{scope}[shift={(6.10,0)}]
  \node at (1.42,1.06) {\textbf{(c)} \(\mathrm{SP}_B\)};
  \node[atom] at (.84,.58) {};
  \node[atom] at (1.76,.58) {};
  \node[atom] at (1.43,0) {};
\end{scope}
\begin{scope}[shift={(9.15,0)}]
  \node at (1.42,1.06) {\textbf{(d)} \(\mathrm M_f\)};
  \node[atom] at (.68,.58) {};
  \node[atom] at (1.18,.58) {};
  \node[atom] at (1.82,.58) {};
  \node[atom] at (2.35,.58) {};
  \node[optional] at (.22,0) {};
  \node[atom] at (.82,0) {};
  \node[atom] at (1.43,0) {};
  \node[atom] at (2.12,0) {};
\end{scope}
\begin{scope}[shift={(12.20,0)}]
  \node at (1.42,1.06) {\textbf{(e)} \(\mathrm M_\infty\)};
  \node at (.58,.58) {$\cdots$};
  \node[atom] at (.94,.58) {};
  \node[atom] at (1.20,.58) {};
  \node[atom] at (1.57,.58) {};
  \node[atom] at (2.02,.58) {};
  \node[atom] at (2.45,.58) {};
  \node[optional] at (.22,0) {};
  \node at (.61,0) {$\cdots$};
  \node[atom] at (.98,0) {};
  \node[atom] at (1.28,0) {};
  \node[atom] at (1.68,0) {};
  \node[atom] at (2.16,0) {};
\end{scope}
\end{tikzpicture}
{\normalsize\begin{spacing}{1}
\caption{Equilibrium-support patterns.}
\label{fig:equilibrium-taxonomy}
\end{spacing}}
\end{figure}
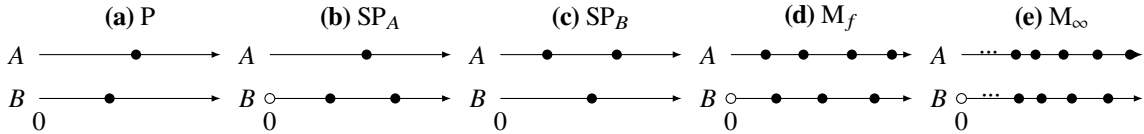

Appendix~\ref{app:proof-general-taxonomy} proves the proposition; the argument has three parts. First, neither contestant can remain inactive for sure: against an inactive rival, a small positive effort beats zero, and every positive effort can be reduced profitably. Nor can both have atoms at zero: replacing zero by tiny effort yields a fixed prize gain whenever the rival is inactive, at vanishing cost.

Second, ratio-dependent prizes imply the \emph{cost-balance identity} \(c\mathbb E[X_A]=\mathbb E[X_B]\equiv\chi\) by averaging first-order conditions. If \(A\) had an atom at zero, \(B\) would have none and \(A\) would earn zero normalized payoff. Independently copying \(B\)'s lottery gives \(A\) expected prize \(1/2\) at cost \(c\chi\). Prize balance then gives \(B\) normalized payoff \(1-2\chi\). Since \(A\) is inactive with positive probability, \(B\) can earn a positive payoff by choosing zero. Hence her equilibrium payoff is positive, so \(\chi<1/2\). Since \(c\leq1\), copying is profitable, a contradiction.

Third, expected payoff has a convergent power-series representation near every positive effort (it is analytic). Thus infinitely many best responses approaching a positive effort would force payoff to be constant, contradicting bounded prizes and linear costs. With finite rival support, clearing denominators in the first-order condition gives a nonzero finite sum of real powers with finitely many positive roots. A pure rival therefore gives the mixer finite support; if either contestant has finite support in a two-sided mixed equilibrium, the other must too.

\section{Outcome Equivalence and Equilibrium Verification}\label{S3}\label{S3:characterization}

We first show why multiple equilibria need not make expected outcomes ambiguous. We then identify the sole pure candidate and construct semi-pure candidates from supporting lines and a required mean ratio. Global deviation tests verify the candidates. For two-sided mixing, the conditions verify proposed distributions; the substantive support restrictions come from Proposition~\ref{thm:equilibrium-taxonomy}.

\subsection{Common Expected Outcomes}\label{S3:outcomes}

\begin{proposition}[Common outcomes]
\label{prop:outcome-equivalence}
All equilibria give each contestant the same expected total effort, prize share, effort cost, and net payoff.
\end{proposition}

By Lemma~\ref{lem:core-reductions}\emph{(ii)}, equilibrium strategies of \(A\) and \(B\) can be paired arbitrarily. Writing these strategy sets as \(\mathcal E_A\) and \(\mathcal E_B\), the equilibrium set is \(\mathcal E_A\times\mathcal E_B\). If either contains distinct strategies, varying their mixing probability continuously generates a continuum of equilibria.

For equilibria \((F_{A,h},F_{B,h})\), \(h\in\{1,2\}\), write \(X_{i,h}\sim F_{i,h}\) and \(\chi_h=c\mathbb E[X_{A,h}]=\mathbb E[X_{B,h}]\). Cost balance applied to the cross-pair \((F_{A,1},F_{B,2})\) gives \(c\mathbb E[X_{A,1}]=\mathbb E[X_{B,2}]\), so \(\chi_1=\chi_2\equiv\chi\). In the transformed constant-sum game, \(A\)'s payoff is \(K_c(x,y)=\mathcal U(x,y)-cx+y\), and \(B\)'s is \(1-K_c(x,y)\). Every equilibrium of this game has the same expected transformed payoff for \(A\); denote it by \(v\). Cost balance gives \(v=\mathbb E[K_c(X_A,X_B)]=\mathbb E[\mathcal U(X_A,X_B)]-\chi+\chi=\mathbb E[\mathcal U(X_A,X_B)]\). Thus expected prize shares are \(\Omega_A=q_0+\Lambda v\) and \(\Omega_B=1-\Omega_A\). In original units, expected total efforts are \((\Lambda\chi/c,\Lambda\chi)\), effort costs \((\Lambda\chi,\Lambda\chi)\), and payoffs \((\Omega_A-\Lambda\chi,\Omega_B-\Lambda\chi)\).

If \(B\) has an atom at zero, indifference and \(F_A(\{0\})=0\) (Proposition~\ref{thm:equilibrium-taxonomy}) give her payoff \(q_0\), so \(\Omega_A=q_J-\Lambda\chi\). Equilibrium selection can change participation and effort dispersion, but not these expected outcomes.

\subsection{Pure Equilibrium}\label{S3:pure-equilibrium}

For normalized total efforts \(x,y\), let \(\mathcal U(x,y)\) be \(A\)'s normalized prize and \(\varphi(t)\) its value at own-to-rival effort ratio \(t\). Given distributions \(F_A,F_B\), normalized deviation payoffs for \(x,y\geq0\) are
\[
\Pi_A(x;F_B)\equiv\int\mathcal U(x,y)\,dF_B(y)-cx,\quad
\Pi_B(y;F_A)\equiv\int[1-\mathcal U(x,y)]\,dF_A(x)-y.
\]
By Proposition~\ref{thm:existence}, pure total-effort deviations suffice: mixed deviations average their payoffs. These averages are finite because prizes are bounded and admissible strategies have finite expected effort.

By Propositions~\ref{thm:existence} and~\ref{thm:equilibrium-taxonomy}, pure equilibria have positive total efforts and uniform vectors. Put \(t=X_B/X_A\), so \(t\) is \(B\)'s own-to-rival effort ratio and \(1/t\) is \(A\)'s. The first-order conditions are \(X_A=\varphi'(t)\) for \(B\) and \(\varphi'(1/t)=cX_B\) for \(A\). Reward symmetry, \(\varphi(t)+\varphi(1/t)=1\), gives \(\varphi'(1/t)=t^2\varphi'(t)\). Substituting \(X_B=tX_A\) yields \(t=c\), since \(t>0\) and monotonicity and nonconstancy give \(\varphi'(t)>0\). Let \(\sigma_c\equiv\varphi'(c)\). These local conditions identify \((X_A,X_B)=(\sigma_c,c\sigma_c)\) as the sole pure candidate; Claim~\ref{prop:general-support} tests all deviations, including inactivity. Below, \(t\) and \(u\) denote \(B\)'s and \(A\)'s own-to-rival effort ratios, respectively.
\begin{claim}[Pure-equilibrium verification]
\label{prop:general-support}
The candidate \((X_A,X_B)=(\sigma_c,c\sigma_c)\) is an equilibrium if and only if, for every \(t,u\geq0\),
\begin{equation}\label{eq:general-pure-support}
\varphi(t)-\sigma_ct\leq\varphi(c)-c\sigma_c,
\qquad
\varphi(u)-c^2\sigma_cu\leq\varphi(1/c)-c\sigma_c.
\end{equation}
\end{claim}

In~\eqref{eq:general-pure-support}, the left-hand sides are \(B\)'s and \(A\)'s normalized deviation payoffs at ratios \(t\) and \(u\); the right-hand sides are their candidate payoffs. The global inequalities include inactivity at \(t=0\) or \(u=0\).

Per-battlefield efforts are \(e_A^*=\Lambda\sigma_c/J\) and \(e_B^*=ce_A^*\). Equation~\eqref{eq:general-pure-support} requires the tangents at \(c\) and \(1/c\) to lie on or above \(\varphi\) at every nonnegative ratio.

\paragraph{Supporting lines and contact sets.}
Fix a positive rival effort. Normalized payoff at effort ratio \(t\) is \(\varphi(t)-\sigma t\), where \(\sigma>0\) equals marginal cost times rival effort. Best-response ratios form \(\Gamma(\sigma)\equiv\arg\max_{t\geq0}\{\varphi(t)-\sigma t\}\). This set is nonempty and compact by continuity and \(\varphi(t)-\sigma t\to-\infty\) as \(t\to\infty\); it is also finite for every \(\sigma>0\) (Appendix~\ref{app:proof-uniqueness-criteria}). The \emph{supporting line} is the lowest line of slope \(\sigma\) globally above \(\varphi\); its \emph{contact set} is \(\Gamma(\sigma)\). Its intercept is the maximal payoff, so all contacts are optimal. At a positive contact \(t\), \(\varphi'(t)=\sigma\), making the line a \emph{supporting tangent}; it is \emph{strict} if it lies strictly above \(\varphi\) away from \(t\).

The \emph{least concave majorant} \(\operatorname{cav}\varphi\) is the pointwise smallest concave function weakly above \(\varphi\). With at least two contacts, set \(t_-=\min\Gamma(\sigma)\) and \(t_+=\max\Gamma(\sigma)\); the segment between them is a \emph{supporting chord}, or an \emph{origin chord} if \(t_-=0\).

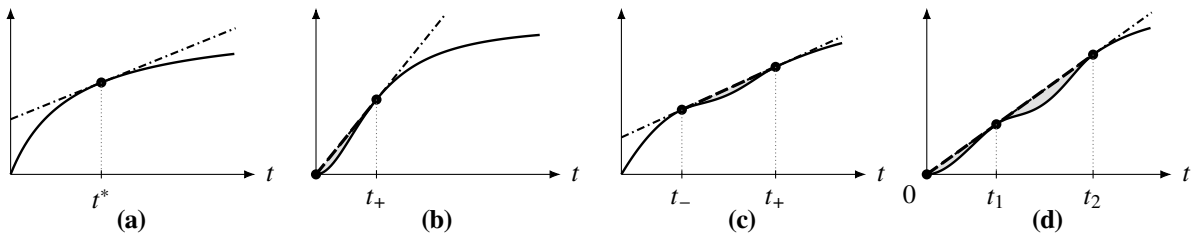
\begin{figure}[H]
\centering
\fitwidth{\begin{tikzpicture}[
  x=.8cm,y=.8cm,
  axis/.style={-{Latex[length=1.5mm]},thin},
  curve/.style={draw=black,line width=.9pt},
  cav/.style={draw=black,line width=1.2pt,dash pattern=on 5pt off 2pt},
  support/.style={draw=black,line width=.8pt,dash pattern=on 3pt off 1.4pt on .7pt off 1.4pt},
  gap/.style={fill=black!12,draw=none},
  guide/.style={draw=black!55,densely dotted,thin},
  contact/.style={circle,fill=black,inner sep=1.35pt},
  every node/.style={font=\footnotesize}
]
\begin{scope}[shift={(0,0)}]
  \draw[axis] (0,0) -- (4.05,0) node[right] {$t$};
  \draw[axis] (0,0) -- (0,2.45);
  \draw[support] (0,0.81) -- (3.72,2.149);
  \draw[curve,domain=0:3.7,samples=80]
    plot (\x,{2.25*\x/(1+\x)});
  \node[contact] at (1.5,1.35) {};
  \draw[guide] (1.5,1.35) -- (1.5,0);
  \draw[thin] (1.5,0.06) -- (1.5,-0.06)
    node[below] {$t^*$};
  \node at (2.0,-0.74) {\textbf{(a)}};
\end{scope}

\begin{scope}[shift={(5.05,0)}]
  \draw[axis] (0,0) -- (4.05,0) node[right] {$t$};
  \draw[axis] (0,0) -- (0,2.45);
  \fill[gap] (0,0) -- (1,1.1)
    plot[domain=1:0,samples=45] (\x,{2.2*\x*\x/(1+\x*\x)}) -- cycle;
  \draw[support] (0,0) -- (2.15,2.365);
  \draw[curve,domain=0:3.7,samples=90]
    plot (\x,{2.2*\x*\x/(1+\x*\x)});
  \draw[cav] (0,0) -- (1,1.1);
  \node[contact] at (0,0) {};
  \node[contact] at (1,1.1) {};
  \draw[guide] (1,1.1) -- (1,0);
  \draw[thin] (1,0.06) -- (1,-0.06)
    node[below] {$t_+$};
  \node at (2.0,-0.74) {\textbf{(b)}};
\end{scope}

\begin{scope}[shift={(10.1,0)}]
  \def\tminus{1}
  \def\tplus{2.55}
  \def\yminus{0.95}
  \def\yplus{1.58}
  \pgfmathsetmacro{\sig}{(\yplus-\yminus)/(\tplus-\tminus)}
  \pgfmathsetmacro{\aint}{\yminus-\sig*\tminus}
  \draw[axis] (0,0) -- (4.05,0) node[right] {$t$};
  \draw[axis] (0,0) -- (0,2.45);
  \fill[gap] (\tminus,\yminus) -- (\tplus,\yplus)
    plot[domain=\tplus:\tminus,samples=60]
      (\x,{\aint+\sig*\x-0.25*(\x-\tminus)^2*(\x-\tplus)^2}) -- cycle;
  \draw[support] (0,{\aint}) -- (3.65,{\aint+\sig*3.65});
  \draw[curve,domain=0:\tminus,samples=40]
    plot (\x,{(2*\yminus-\sig)*\x+(\sig-\yminus)*\x*\x});
  \draw[curve,domain=\tminus:\tplus,samples=60]
    plot (\x,{\aint+\sig*\x-0.25*(\x-\tminus)^2*(\x-\tplus)^2});
  \draw[curve,domain=\tplus:3.65,samples=40]
    plot (\x,{\yplus+\sig*(\x-\tplus)-0.08*(\x-\tplus)^2});
  \draw[cav] (\tminus,\yminus) -- (\tplus,\yplus);
  \node[contact] at (\tminus,\yminus) {};
  \node[contact] at (\tplus,\yplus) {};
  \draw[guide] (\tminus,\yminus) -- (\tminus,0);
  \draw[guide] (\tplus,\yplus) -- (\tplus,0);
  \draw[thin] (\tminus,0.06) -- (\tminus,-0.06)
    node[below] {$t_-$};
  \draw[thin] (\tplus,0.06) -- (\tplus,-0.06)
    node[below] {$t_+$};
  \node at (2.0,-0.74) {\textbf{(c)}};
\end{scope}

\begin{scope}[shift={(15.15,0)}]
  \def\tone{1.15}
  \def\ttwo{2.75}
  \def\sig{0.64}
  \pgfmathsetmacro{\ytwo}{\sig*\ttwo}
  \draw[axis] (0,0) -- (4.05,0) node[right] {$t$};
  \draw[axis] (0,0) -- (0,2.45);
  \fill[gap] (0,0) -- (\tone,{\sig*\tone})
    plot[domain=\tone:0,samples=50]
      (\x,{\sig*\x-(\sig/(\tone*\tone))*\x*(\x-\tone)^2}) -- cycle;
  \fill[gap] (\tone,{\sig*\tone}) -- (\ttwo,\ytwo)
    plot[domain=\ttwo:\tone,samples=60]
      (\x,{\sig*\x-0.5*(\x-\tone)^2*(\x-\ttwo)^2}) -- cycle;
  \draw[support] (0,0) -- (3.72,{\sig*3.72});
  \draw[curve,domain=0:\tone,samples=50]
    plot (\x,{\sig*\x-(\sig/(\tone*\tone))*\x*(\x-\tone)^2});
  \draw[curve,domain=\tone:\ttwo,samples=60]
    plot (\x,{\sig*\x-0.5*(\x-\tone)^2*(\x-\ttwo)^2});
  \draw[curve,domain=\ttwo:3.7,samples=40]
    plot (\x,{\ytwo+\sig*(1-exp(-(\x-\ttwo)))});
  \draw[cav] (0,0) -- (\ttwo,\ytwo);
  \node[contact] at (0,0) {};
  \node[contact] at (\tone,{\sig*\tone}) {};
  \node[contact] at (\ttwo,\ytwo) {};
  \draw[guide] (\tone,{\sig*\tone}) -- (\tone,0);
  \draw[guide] (\ttwo,\ytwo) -- (\ttwo,0);
  \draw[thin] (\tone,0.06) -- (\tone,-0.06)
    node[below] {$t_1$};
  \draw[thin] (\ttwo,0.06) -- (\ttwo,-0.06)
    node[below] {$t_2$};
  \node[below left] at (0,0) {$0$};
  \node at (2.0,-0.74) {\textbf{(d)}};
\end{scope}
\end{tikzpicture}}
{\normalsize\begin{spacing}{1}
\caption{Best-response geometry (schematic).}
\label{fig:supporting-geometry}
\end{spacing}}
\end{figure}

\afterfigurespace

In Figure~\ref{fig:supporting-geometry}, the solid curve is \(\varphi\), long dashes show \(\operatorname{cav}\varphi\) where it exceeds \(\varphi\), and the dash-dotted line is the supporting line; shading marks the gap, filled dots the contacts, and dotted guides their ratios. Panels~(a)--(d) have contact sets \(\{t^*\}\), \(\{0,t_+\}\), \(\{t_-,t_+\}\), and \(\{0,t_1,t_2\}\), respectively. In panel~(d), \(t_-=0\), \(t_+=t_2\), and the supporting chord is also tangent at the interior contact \(t_1\).

\subsection{One-Sided Mixing}\label{S3:semi-pure-equilibrium}

The mixer's ratio is own effort divided by the pure rival's. If \(A\) fixes normalized total effort \(x_A^\circ>0\), \(B\)'s ratio \(t\geq0\) gives effort \(x_A^\circ t\), payoff \(\varphi(t)-x_A^\circ t\), and best-response ratios \(\Gamma(x_A^\circ)\). If \(B\) fixes \(x_B^\circ>0\), \(A\)'s ratio \(\rho\geq0\) gives effort \(x_B^\circ\rho\), payoff \(\varphi(\rho)-cx_B^\circ\rho\), and best-response ratios \(\Gamma(cx_B^\circ)\).

For an \(A\)-pure candidate, set \(T_B\equiv X_B/x_A^\circ\), with distribution \(F_{T_B}\) and inactivity probability \(\zeta_B\equiv F_{T_B}(\{0\})\). Define \(T_A\equiv X_A/x_B^\circ\), \(F_{T_A}\), and \(\zeta_A\) analogously for a \(B\)-pure candidate. Below, \(t\) and \(\rho\) are realizations of \(T_B\) and \(T_A\); \(F_B\) and \(F_A\) are the distributions of \(x_A^\circ T_B\) and \(x_B^\circ T_A\), respectively.

A lottery over contacts makes the mixer indifferent among all played ratios. Its probabilities must also satisfy the equal-expected-cost identity. The \emph{mean equality} follows from \(c\mathbb E[X_A]=\mathbb E[X_B]\): when \(A\) is pure, \(cx_A^\circ=x_A^\circ\mathbb E[T_B]\), so \(\mathbb E[T_B]=c\); when \(B\) is pure, \(cx_B^\circ\mathbb E[T_A]=x_B^\circ\), so \(\mathbb E[T_A]=1/c\). For example, an \(A\)-pure lottery assigning probability \(p_k\) to contact \(t_k\) satisfies \(\sum_kp_kt_k=c\), including any mass at zero. Nontrivial mixing requires positive probability below and above the mean.

The mean restriction is necessary, but the pure contestant must still weakly prefer her proposed effort to every deviation. Hold the mixer's distribution fixed and scale the pure contestant's effort by \(z>0\). At \(zx_A^\circ\) or \(zx_B^\circ\), normalized payoffs are
\[
\Pi_A(zx_A^\circ;F_B)=\zeta_B+\int_{t>0}\varphi(z/t)\,dF_{T_B}(t)-cx_A^\circ z,\quad
\Pi_B(zx_B^\circ;F_A)=\zeta_A+\int_{\rho>0}\varphi(z/\rho)\,dF_{T_A}(\rho)-x_B^\circ z.
\]
The integrals average rewards over positive rival ratios; the \(\zeta_i\) terms represent inactivity and are constant for \(z>0\). At zero effort, \(\Pi_A(0;F_B)=\zeta_B/2\) and \(\Pi_B(0;F_A)=\zeta_A/2\). The proposed effort must maximize \(\Pi_i\) over all nonnegative efforts, including inactivity, and Proposition~\ref{thm:equilibrium-taxonomy} requires \(\zeta_A=0\). Thus
\begin{equation}\label{eq:semi-pure-requirements}
\begin{aligned}
\text{\(A\) pure:}\quad&
\underbrace{\mathbb E[T_B]=c}_{\text{mean equality}},
\quad
\underbrace{x_A^\circ\in\arg\max_{x\geq0}\Pi_A(x;F_B)}_{\text{pure contestant's condition}},
\quad
\underbrace{\operatorname{supp}F_{T_B}\subseteq\Gamma(x_A^\circ)}_{\text{mixer's conditions}};\\[6pt]
\text{\(B\) pure:}\quad&
\underbrace{\mathbb E[T_A]=1/c}_{\text{mean equality}},
\quad
\underbrace{x_B^\circ\in\arg\max_{y\geq0}\Pi_B(y;F_A)}_{\text{pure contestant's condition}},
\quad
\underbrace{\operatorname{supp}F_{T_A}\subseteq\Gamma(cx_B^\circ),\ \zeta_A=0}_{\text{mixer's conditions}}.
\end{aligned}
\end{equation}

\begin{claim}[Semi-pure verification]
\label{prop:general-semi-pure}
A profile is a semi-pure equilibrium if and only if the mixer's ratio distribution is not a point mass and satisfies the appropriate row of \eqref{eq:semi-pure-requirements}.
\end{claim}

Necessity follows directly. Conversely, contact support verifies the mixer's best response, and the global conditions verify the pure contestant's, including inactivity.

\paragraph{Constructing semi-pure equilibria.}

Claim~\ref{prop:general-semi-pure} yields an \emph{exhaustive candidate construction and verification procedure for semi-pure equilibria}. Write \(g=\operatorname{cav}\varphi\), and let \(m\) denote the mixer's required mean effort ratio: \(m=c\) when \(A\) is pure and \(m=1/c\) when \(B\) is pure. Figure~\ref{fig:mean-selected-effort} illustrates the procedure: solid curves show \(\varphi\), long-dashed segments show \(g\), dash-dotted lines show selected supporting lines, and filled dots are contacts.

\begin{figure}[H]
\centering
\scalebox{.95}{\begin{tikzpicture}[yscale=.80,
  x=1cm,y=1cm,
  axis/.style={-{Latex[length=1.5mm]},thin},
  curve/.style={black,line width=.8pt},
  cav/.style={black,line width=1.25pt,dash pattern=on 4pt off 1.7pt},
  support/.style={black,line width=.7pt,dash pattern=on 3pt off 1.3pt on .7pt off 1.3pt},
  other/.style={black!55,line width=.7pt,dash pattern=on 3pt off 1.3pt on .7pt off 1.3pt},
  guide/.style={black!50,densely dotted,thin},
  contact/.style={circle,fill=black,inner sep=1.4pt},
  meanpoint/.style={circle,draw=black,fill=white,line width=.6pt,inner sep=1.6pt},
  every node/.style={font=\footnotesize,inner sep=1.5pt}
]
\node at (2.15,-.64) {\textbf{(a)}};
\node at (7.60,-.64) {\textbf{(b)}};
\node at (13.05,-.64) {\textbf{(c)}};

\begin{scope}[x=.66cm,y=.82cm]
  \draw[axis] (0,0) -- (6.55,0) node[right] {$t$};
  \draw[axis] (0,0) -- (0,3.62);
  \fill[black!10] (.6,.65) -- (1.5,1.325)
    plot[domain=1.5:.6,samples=45]
      (\x,{.2+.75*\x-3*(\x-.6)^2*(\x-1.5)^2}) -- cycle;
  \fill[black!10] (2.8,1.975) -- (4.1,2.3)
    plot[domain=4.1:2.8,samples=50]
      (\x,{1.275+.25*\x-.5*(\x-2.8)^2*(\x-4.1)^2}) -- cycle;
  \fill[black!10] (4.1,2.3) -- (5.7,2.7)
    plot[domain=5.7:4.1,samples=50]
      (\x,{1.275+.25*\x-.25*(\x-4.1)^2*(\x-5.7)^2}) -- cycle;
  \draw[other] (0,.2) -- (3.93,3.1475)
    node[above left,text=black!65] {$\ell_1$};
  \draw[support] (0,1.275) -- (6.38,2.87)
    node[above right] {$\ell$};
  \draw[curve,domain=0:.6,samples=35]
    plot (\x,{(17/12)*\x-(5/9)*\x*\x});
  \draw[curve,domain=.6:1.5,samples=50]
    plot (\x,{.2+.75*\x-3*(\x-.6)^2*(\x-1.5)^2});
  \draw[curve,domain=1.5:2.8,samples=50]
    plot (\x,{1.325+.75*(\x-1.5)-(5/26)*(\x-1.5)^2});
  \draw[curve,domain=2.8:4.1,samples=60]
    plot (\x,{1.275+.25*\x-.5*(\x-2.8)^2*(\x-4.1)^2});
  \draw[curve,domain=4.1:5.7,samples=60]
    plot (\x,{1.275+.25*\x-.25*(\x-4.1)^2*(\x-5.7)^2});
  \draw[curve,domain=5.7:6.35,samples=35]
    plot (\x,{2.7+.25*(1-exp(-(\x-5.7)))});
  \draw[cav] (.6,.65) -- (1.5,1.325);
  \draw[cav] (2.8,1.975) -- (5.7,2.7);
  \foreach \xx/\yy in {.6/.65,1.5/1.325,2.8/1.975,4.1/2.3,5.7/2.7}
    \node[contact] at (\xx,\yy) {};
  \draw[guide] (2.8,1.975) -- (2.8,0) node[below,text=black] {$t_1$};
  \draw[guide] (5.7,2.7) -- (5.7,0) node[below,text=black] {$t_3$};
  \draw[guide] (3.7,2.2) -- (3.7,0) node[below,text=black] {$m$};
  \node[meanpoint] at (3.7,2.2) {};

\end{scope}

\begin{scope}[shift={(5.5,0)}]
  \draw[axis] (0,0) -- (4.12,0) node[right] {$t$};
  \draw[axis] (0,0) -- (0,3.03);
  \fill[black!10] (.5,.8) -- (1.8,1.32)
    plot[domain=1.8:.5,samples=50]
      (\x,{.6+.4*\x-.8*(\x-.5)^2*(\x-1.8)^2}) -- cycle;
  \fill[black!10] (1.8,1.32) -- (3.4,1.96)
    plot[domain=3.4:1.8,samples=50]
      (\x,{.6+.4*\x-.4*(\x-1.8)^2*(\x-3.4)^2}) -- cycle;
  \draw[support] (.08,.632) -- (3.94,2.176) node[above right] {$\ell$};
  \draw[curve,domain=.15:.5,samples=25]
    plot (\x,{.8+.4*(\x-.5)-(4/13)*(\x-.5)^2});
  \draw[curve,domain=.5:1.8,samples=50]
    plot (\x,{.6+.4*\x-.8*(\x-.5)^2*(\x-1.8)^2});
  \draw[curve,domain=1.8:3.4,samples=50]
    plot (\x,{.6+.4*\x-.4*(\x-1.8)^2*(\x-3.4)^2});
  \draw[curve,domain=3.4:3.9,samples=25]
    plot (\x,{1.96+.4*(1-exp(-(\x-3.4)))});
  \draw[cav] (.5,.8) -- (3.4,1.96);
  \foreach \xx/\yy/\lab in {.5/.8/t_1,1.8/1.32/t_2,3.4/1.96/t_3} {
    \draw[guide] (\xx,\yy) -- (\xx,0) node[below,text=black] {$\lab$};
    \node[contact] at (\xx,\yy) {};
  }
  \draw[guide] (1.4,1.16) -- (1.4,0) node[below,text=black] {$m$};
  \node[meanpoint] at (1.4,1.16) {};

\end{scope}

\begin{scope}[shift={(11.25,0)},x=3.35cm,y=2.7cm]
  \fill[black!5] (0,0) -- (1,0) -- (0,1) -- cycle;
  \draw[black!40,thin] (1,0) -- (0,1);
  \draw[axis] (0,0) -- (1.12,0) node[right] {$p_2$};
  \draw[axis] (0,0) -- (0,1.08) node[right] {$p_3$};
  \node[below left] at (0,0) {$0$};
  \node[below] at (1,0) {$1$};
  \node[left] at (0,1) {$1$};

  \draw[black,line width=1.8pt] ({.9/1.3},0) -- (0,{.9/2.9});
  \node[contact] at ({.9/1.3},0) {};
  \node[contact] at (0,{.9/2.9}) {};
  \node[contact] at ({.45/1.3},{.45/2.9}) {};

\end{scope}

\end{tikzpicture}}
{\normalsize\begin{spacing}{1}
\caption{Mean selection and uniqueness of the candidate pure effort (schematic).}
\label{fig:mean-selected-effort}
\end{spacing}}
\end{figure}
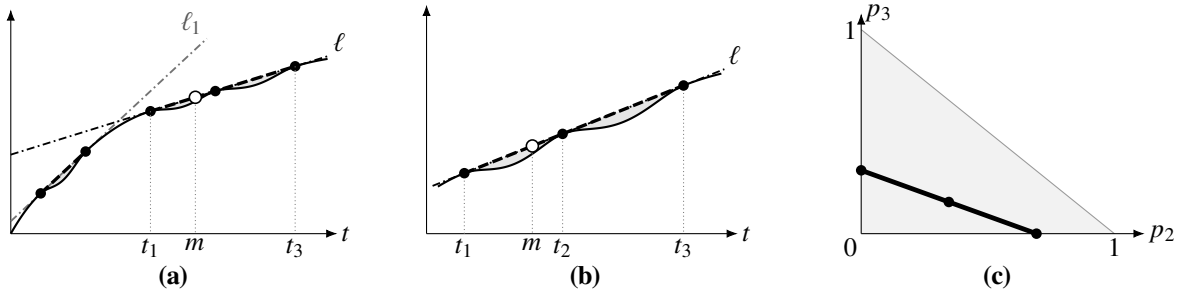

\afterfigurespace

\begin{enumerate}
\setcounter{enumi}{-1}
\item \emph{Fix the pure contestant.} Choosing \(A\) or \(B\) fixes the required mean \(m\).

\item \emph{Select the supporting chord (panel~(a)).} Find a supporting line with contact ratios strictly on both sides of \(m\). Both \(\ell_1\) and \(\ell\) support \(\varphi\), but all contacts of \(\ell_1\) lie below \(m\), ruling out the required mean. The contacts of \(\ell\) straddle \(m\). If no such line exists, there is no semi-pure equilibrium with the chosen pure contestant.

\item \emph{Determine the candidate pure effort (panel~(b)).} Panel~(b) enlarges the selected chord from panel~(a). Because such a line lies weakly above \(\varphi\) and touches it at contacts \(t'<m<t''\), \(g\) coincides with it on \([t',t'']\). Hence \(g\) is a straight line near \(m\), every such line has slope \(\sigma=g'(m)\), and this slope uniquely fixes \(x_A^\circ=g'(c)\) when \(A\) is pure, or \(x_B^\circ=g'(1/c)/c\) when \(B\) is pure. The open circle marks \((m,g(m))\), which need not be a contact.

\item \emph{Assign contact probabilities with the required mean (panel~(c)).} Consider every nondegenerate lottery on \(\Gamma(\sigma)\) with mean \(m\). In the illustrated three-contact case, \(t_1<m<t_2<t_3\). A point \((p_2,p_3)\) in the triangle specifies probabilities \(p_1=1-p_2-p_3\), \(p_2\), and \(p_3\) on these contacts. The mean restriction \((t_2-t_1)p_2+(t_3-t_1)p_3=m-t_1\) selects the thick segment. The horizontal-axis endpoint mixes over \(\{t_1,t_2\}\), the vertical-axis endpoint over \(\{t_1,t_3\}\), and interior points use all three contacts. These different lotteries share the same mean and hence the same candidate pure effort. When \(A\) mixes, exclude zero from the support.

\item \emph{Verify the pure contestant's deviations.} Keep exactly the lotteries satisfying the global best-response conditions, including inactivity, in~\eqref{eq:semi-pure-requirements}. This identifies all semi-pure equilibria, even when the candidate family is a continuum. \end{enumerate}

\subsection{Two-Sided Mixing}\label{S3:two-sided}

Since normalized prizes are at most one, efforts above \(1/c\) for \(A\) or one for \(B\) are dominated by inactivity. Thus take \(F_A\) and \(F_B\) supported on \([0,1/c]\) and \([0,1]\), respectively.

To verify a proposed pair \((F_A,F_B)\), find constants \(v_A,v_B\) such that
\begin{align}
\Pi_A(x;F_B)&\leq v_A &&\text{for every }x\geq0,
\label{eq:general-complementarity-A}\\
\Pi_B(y;F_A)&\leq v_B &&\text{for every }y\geq0,
\label{eq:general-complementarity-B}
\end{align}
with equality \(F_A\)- and \(F_B\)-almost surely, respectively. These are the mutual best-response conditions, written to include all positive-effort deviations and inactivity. If neither distribution is a point mass, they verify a two-sided mixed equilibrium; conversely, every such equilibrium satisfies them at its equilibrium payoffs.

The substantive restrictions for arbitrary schedules are those in Proposition~\ref{thm:equilibrium-taxonomy}: only the higher-cost contestant may be inactive with positive probability, and positive supports are either both finite or both countably infinite, accumulating only at zero. The verification conditions themselves do not identify which support pattern occurs or solve for the distributions.

\section{Uniqueness and Multiplicity across Battlefield Counts}\label{S4}

Common outcomes do not imply unique strategies. We now show how the battlefield count limits multiplicity and mixing.

\subsection{The Battlefield-Count Frontier}\label{S4:counts}

\begin{theorem}[Uniqueness through six battlefields]
\label{thm:lowcount-refinements}
If \(J\leq6\), equilibrium is unique and either pure or semi-pure, with any mixer using exactly two effort levels.
\end{theorem}

FKL (Theorem~1 and Propositions~3--4) establish that equilibrium is unique and pure for every admissible schedule and cost ratio whenever \(0<r\leq r_J^{\mathrm{all}}\equiv1/\lceil J/2\rceil\). This bound is sharp: for every \(J\geq3\) and \(r>r_J^{\mathrm{all}}\), majority rule has no pure equilibrium for all sufficiently small \(c\). Lemma~\ref{lem:purity-bound} in Appendix~\ref{app:proof-battlefield-count} records these results. Table~\ref{tab:equilibrium-by-count} summarizes equilibrium forms and the sharp purity bounds.

\suppressfloats[t]
\begin{table}[!htbp]
\centering
\begingroup\small\compactdisplays\setstretch{1.15}
\setlength{\tabcolsep}{5pt}
\renewcommand{\arraystretch}{1.18}
\begin{tabular}{@{}*{7}{c}@{}}
\hline
\(J\) & \multicolumn{5}{c}{Possible equilibrium forms} & Pure for all schedules and costs \\
 & \(\mathrm P\) & \(\mathrm{SP}_A\) & \(\mathrm{SP}_B\) & \(\mathrm M_f\) & \(\mathrm M_\infty\) & \\
\hline
\(1,2\) & \(\checkmark\) & & & & & \(0<r\leq1\) \\
\(3,4\) & \(\checkmark\) & \(\checkmark\) & & & & \(0<r\leq1/2\) \\
5 & \(\checkmark\) & \(\checkmark\) & & & & \(0<r\leq1/3\) \\
6 & \(\checkmark\) & \(\checkmark\) & \(\checkmark\) & & & \(0<r\leq1/3\) \\
\(\ge7\) & \(\checkmark\) & \(\checkmark\) & \(\checkmark\) & \(\checkmark\) & \(\checkmark\) & \(0<r\leq1/\lceil J/2\rceil\) \\
\hline
\end{tabular}
\caption{Attainable equilibrium forms and sharp schedule-robust purity bounds.}\label{tab:equilibrium-by-count}
\par\smallskip
\noindent
\begin{minipage}{\textwidth}
Checkmarks indicate attainability as the admissible schedule, \(c\), and \(r\) vary, not coexistence at fixed parameters.
\end{minipage}
\endgroup
\end{table}

The proof uses supporting-line geometry: through six battlefields, at most one chord arises, with exactly two contacts. A strict best-response argument fixes one contestant's effort in every equilibrium, and equal expected costs determine the rival's mixing probabilities. Only the higher-cost contestant can mix through five battlefields; either can at six. Appendix~\ref{app:proof-battlefield-count} gives the proof.

\subsection{Uniqueness and Coexistence}\label{S4:uniqueness}

Section~\ref{S3} supplies candidate constructions and verification conditions. We now study how equilibrium forms can coexist and when multiplicity arises. Interchangeability lets us pair strategies from different equilibria, so best-response restrictions from one equilibrium constrain all others. Together with contact geometry, these restrictions also provide tests for uniqueness. As above, \(\mathrm P\) denotes pure equilibrium and \(\mathrm{SP}_i\) semi-pure equilibrium with \(i\) pure; write \(\mathrm M\) for two-sided mixing.

\begin{corollary}[Restrictions on coexistence of equilibrium types]\label{cor:type-coexistence}
For fixed model parameters, the number of distinct equilibrium types can only be one, two, or four, never three. If exactly two types coexist, the pair is one of \(\{\mathrm P,\mathrm{SP}_A\}\), \(\{\mathrm P,\mathrm{SP}_B\}\), \(\{\mathrm{SP}_A,\mathrm M\}\), \(\{\mathrm{SP}_B,\mathrm M\}\).
At equal costs, the type set is \(\{\mathrm P\}\), \(\{\mathrm M\}\), or all four.
\end{corollary}

By Lemma~\ref{lem:core-reductions}\emph{(ii)}, equilibrium strategies can be paired across contestants. Each contestant's equilibrium-strategy set contains only point masses, only nondegenerate distributions, or both. Their combinations give one, two, or four types, never three; exactly one set containing both categories yields a listed pair. At \(c=1\), symmetry makes the two sets identical.

\begin{corollary}[Infinite support and coexistence]\label{cor:common-support-sets}
Fix the model parameters.
If one equilibrium has infinite support, every equilibrium is two-sided mixed with countably infinite positive supports. Consequently, coexistence of two-sided mixing with a pure or semi-pure equilibrium requires finite supports.
\end{corollary}
Interchangeability puts every equilibrium's positive efforts in common discrete best-response sets. If one equilibrium has finite supports, the first-order conditions against its strategies make these sets finite; hence every equilibrium has finite supports. Proposition~\ref{thm:equilibrium-taxonomy} then gives the conclusion. Appendix~\ref{app:finite-support-criteria} gives the proof.

A single type need not imply uniqueness: mixing distinct equilibrium strategies yields a continuum of equilibria. To exclude two-sided mixing, it suffices that either contestant's pure effort be her unique best response to the rival's strategy in one verified equilibrium, with the comparison including inactivity. Apply either tangent test in Claim~\ref{prop:general-support}, or the pure contestant's test in Claim~\ref{prop:general-semi-pure}, with strict inequalities at every alternative effort. This fixes one contestant's strategy, but need not fix the other's; excluding two-sided mixing does not by itself establish equilibrium uniqueness.

If both tangent inequalities in Claim~\ref{prop:general-support} are strict away from \(c\) and \(1/c\), respectively, the same argument fixes both efforts and shows that equilibrium is unique and pure. Proportional prizes supply this benchmark for every finite \(J\): \(\varphi(t)=t^r/(1+t^r)\) is strictly concave for \(r\leq1\).

\begingroup
\paragraph{The pure-contestant best-response test.} We now give a sufficient test for all lotteries over the contacts by checking \(A\)'s deviations against each positive contact separately. Take the \(A\)-pure candidate \(X_A=\sigma>0\) from Sections~\ref{S3:pure-equilibrium}--\ref{S3:semi-pure-equilibrium}: \(\sigma=\varphi'(c)\) in the pure case, or \(\sigma=(\operatorname{cav}\varphi)'(c)\) in the semi-pure case. Let \(T_B=X_B/\sigma\) have distribution \(F_{T_B}\), supported on \(\Gamma(\sigma)\) with mean \(c\), and let \(F_B\) be the induced effort distribution. Contact support verifies \(B\)'s best response; the mean condition enforces cost balance. It remains to verify that \(A\) cannot gain by changing effort.

Against a positive contact \(u\), \(B\)'s effort is \(\sigma u\), so \(A\)'s proposed effort has ratio \(1/u\). Tangency and reward symmetry give \(\varphi'(u)=\sigma\) and \(\varphi'(1/u)=u^2\sigma\). The test requires the tangent at this reciprocal ratio to support the entire reward curve for every positive contact \(u\):
\[
\varphi(z)-u^2\sigma z\leq\varphi(1/u)-u\sigma
\qquad\text{for every }z\geq0.
\]
If \(A\) changes effort to \(s\sigma>0\), her ratio becomes \(s/u\). Substituting \(z=s/u\) bounds the reward change against contact \(u\) by \(u\sigma(s-1)\). Averaging gives \(\sigma(s-1)\mathbb E[T_B]=c\sigma(s-1)\), exactly her cost change; an inactive rival contributes no reward change.

Inactivity is checked separately: the inequalities at \(z=0\) give \(\Pi_A(\sigma;F_B)\geq\zeta_B\geq\zeta_B/2=\Pi_A(0;F_B)\), where \(\zeta_B=F_{T_B}(\{0\})\). Thus every mean-feasible contact lottery generates an equilibrium. These weak inequalities verify optimality, not uniqueness of \(A\)'s best response. If this sufficient test fails, a candidate may still pass the lottery-specific \(\Pi_A\) test.
\par\endgroup

\begingroup
\paragraph{Ruling out two-sided mixing.} Call the test \emph{strict} if, for every positive contact \(u\), the preceding inequality is strict except at \(z=1/u\): each reciprocal tangent touches the curve only there. Since \(c>0\), every mean-feasible lottery assigns positive probability to positive contacts, so averaging retains strictness against every other positive effort. At \(z=0\), the same strict inequalities give \(\Pi_A(\sigma;F_B)>\zeta_B\geq\Pi_A(0;F_B)\). Thus \(\sigma\) is \(A\)'s unique best response, including inactivity. Interchangeability then fixes \(A\)'s effort in every equilibrium, ruling out two-sided mixing. In any remaining equilibrium, \(B\) must use contacts with mean \(c\). The verified family therefore contains all equilibria. Appendix~\ref{app:proof-uniqueness-criteria} gives the proof.
\par\endgroup

\begin{proposition}[Equilibria supported by a finite contact set]
\label{prop:finite-contact}
Suppose there exists \(\sigma>0\) such that \(\Gamma(\sigma)\) satisfies the strict pure-contestant best-response test and supports a distribution with mean \(c\). Then the complete equilibrium set consists exactly of the profiles \(X_A=\sigma\) and \(X_B=\sigma T_B\), where \(F_{T_B}\) ranges over all distributions supported on \(\Gamma(\sigma)\) with mean \(c\). The profile is pure when \(F_{T_B}\) is a point mass and semi-pure with \(A\) pure otherwise. The role-reversed statement holds for a \(B\)-pure candidate: if \(\Gamma(\sigma)\setminus\{0\}\) satisfies the strict test and supports a distribution with mean \(1/c\), the complete equilibrium set consists of the profiles \(X_B=\sigma/c\) and \(X_A=(\sigma/c)T_A\), where \(F_{T_A}\) ranges over all distributions supported on \(\Gamma(\sigma)\setminus\{0\}\) with mean \(1/c\).
\end{proposition}

Once Proposition~\ref{prop:finite-contact} fixes \(A\)'s effort, uniqueness depends only on \(B\)'s contact probabilities. A single contact, necessarily at \(c\), gives a unique pure equilibrium. Two contacts \(t_-<c<t_+\) give the unique semi-pure lottery, with \(\Pr(T_B=t_+)=(c-t_-)/(t_+-t_-)\). With at least three contacts and \(\min\Gamma(\sigma)<c<\max\Gamma(\sigma)\), the total-mass and mean constraints leave at least one probability free, generating a continuum. A mean equal to an extreme contact instead forces the point mass at \(c\). The same counting applies to \(A\)'s contact probabilities in the role-reversed case, with mean \(1/c\).

\subsection{Examples and Sharpness}\label{S4:examples}

We construct prize schedules that illustrate the equilibrium forms and show that the battlefield-count frontier is sharp. After reviewing small-count benchmarks, we give a six-battlefield example in which only the lower-cost contestant mixes and seven-battlefield examples with either a continuum of equilibria or a unique finite-support equilibrium in which both contestants mix. We first introduce degree elevation, which extends these examples to larger counts without changing their total-effort equilibrium sets.

To define the procedure, take coefficients \(\boldsymbol a=(a_0,\ldots,a_K)\) and write
\(\mathcal B_K[\boldsymbol a](p)=\sum_{n=0}^{K}a_n\binom{K}{n}p^n(1-p)^{K-n}\).
If \(W\sim\operatorname{Binomial}(K,p)\), then \(\mathcal B_K[\boldsymbol a](p)=\mathbb E[a_W]\).
For a \(J\)-battlefield schedule, define its \emph{degree elevation} \(\boldsymbol q^+\) by
\[
q_0^+=q_0,\qquad
q_n^+=\frac{nq_{n-1}+(J+1-n)q_n}{J+1}
\quad(n=1,\ldots,J),\qquad
q_{J+1}^+=q_J.
\]
\begin{lemma}[Degree elevation]\label{lem:cross-count-identities}
For every admissible \(J\)-battlefield schedule \(\boldsymbol q\), its degree elevation
\(\boldsymbol q^+\) is admissible and
\(\mathcal B_{J+1}[\boldsymbol q^+]=\mathcal B_J[\boldsymbol q]\).
\end{lemma}

By Proposition~\ref{thm:existence}, repeated elevation therefore extends each example to every larger count for fixed \(r\) and \(c\), leaving total-effort strategies unchanged and rescaling per-battlefield effort inversely with the count. Appendix~\ref{app:cross-count-identities} proves the lemma.

One and two battlefields induce the same scalar Tullock contest and yield the same normalized total-effort equilibrium across counts. For three-battlefield majority at \(r=1\), \(\varphi(t)=t^2(3+t)/(1+t)^3\). The origin chord has positive contact \(t_+=\sqrt7-2\), so \(B\) mixes between inactivity and one positive effort when \(c<t_+\). Degree elevation extends this example to four and five battlefields. Online Appendix~\ref{OA:small-counts} gives general three- and four-battlefield cutoffs and a five-battlefield example with mixing between two positive efforts.

\paragraph{Six battlefields.}\label{S4:six}
Equilibrium remains unique, but the lower-cost contestant can now be the sole mixer.

\begin{example}[Six battlefields: the lower-cost contestant mixes]\label{ex:six-lower-cost-mixing}
Let \(J=6\), \(r=1\), \(c=1/2\), \(\Lambda=1\), and \(\widetilde{\boldsymbol q}=(0,\tfrac12,\tfrac12,\tfrac12,\tfrac12,\tfrac12,1)\). The unique chord has contacts \(t_-\approx1.06518\) and \(t_+\approx3.06644\), slope \(\sigma\approx0.0444009\), and mixing probability \(\lambda=(2-t_-)/(t_+-t_-)\approx0.46712\). The unique equilibrium distributions of normalized total effort are \(F_B=\delta_{2\sigma}, F_A=(1-\lambda)\delta_{2\sigma t_-}+\lambda\delta_{2\sigma t_+}\).

\end{example}
 Only the first and last victories raise the prize, creating distinct incentives to avoid losing every battlefield or to win them all. Online Appendix~\ref{app:six-lower-cost-mixing} verifies the construction.

\Needspace{8\baselineskip}
\paragraph{Seven battlefields: multiplicity and two-sided mixing.}\label{S4:complexity}

\begin{example}[Three contacts and a continuum of equilibria]
\label{ex:seven-multiplicity}
Let \(J=7\), \(r=1\), and
\[
\widetilde q_0=0,\quad
\widetilde q_1=\frac{6101}{75852},\quad
\widetilde q_2=\frac{538607}{2730672},\quad
\widetilde q_3=\frac{10209}{33712},\qquad
\widetilde q_{7-n}=1-\widetilde q_n\quad(n=0,1,2,3).
\]
Then \(\Gamma(\sigma)=\{0,1/9,1/4\}\) satisfies the strict pure-contestant best-response test for \(\sigma=612/1075\).
\end{example}

The schedule is admissible because \(0<\widetilde q_1<\widetilde q_2<\widetilde q_3<1/2\). Online Appendix~\ref{app:seven-multiplicity-proof} verifies the contacts and the strict best-response test. For \(0<c\leq1/4\), Proposition~\ref{prop:finite-contact} gives the complete equilibrium family: \(A\) chooses \(\sigma\), while \(B\) assigns probabilities \((p_0,p_1,p_2)\) to \((0,\sigma/9,\sigma/4)\), with \(p_j\geq0\), \(p_0+p_1+p_2=1\), and \(p_1/9+p_2/4=c\). For every \(c\in(0,1/4)\), the feasible probability vectors form a nondegenerate line segment.

At \(c=1/9\), Figure~\ref{fig:coexistence-uniqueness} shows the scaled gap \(10^5[\sigma t-\varphi(t)]\) and three contacts in panel~(a), and supporting tangents \(\ell_4\) and \(\ell_9\) at reciprocal ratios \(z=4\) and \(z=9\) in panel~(b). The lotteries \(L_1=\tfrac59\delta_0+\tfrac49\delta_{1/4}\) (solid) and \(L_0=\delta_{1/9}\) (dashed) are over \(B\)'s ratio \(T_B=X_B/\sigma\). Panel~(c) plots \(A\)'s payoff at effort \(s\sigma\), uniquely maximized at \(s=1\). Under \(L_1\), inactivity yields \(5/18\), versus a payoff limit of \(5/9\) as positive effort decreases to zero. Panel~(d) shows the mean-feasible segment joining \(L_1\) and \(L_0\).

\begin{figure}[H]
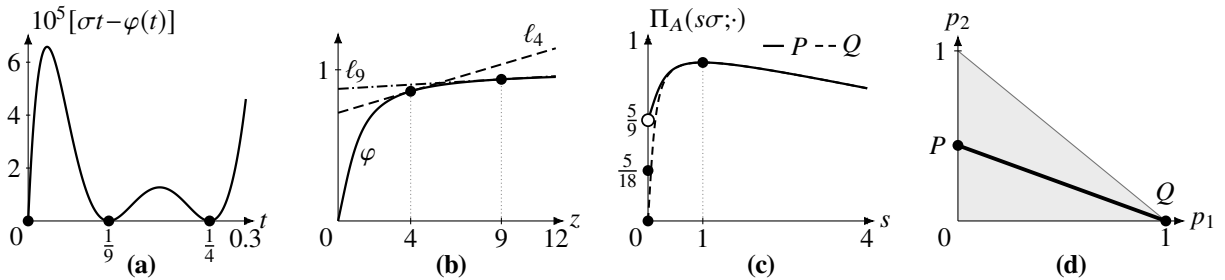

\centering
\fitwidth{
}
{\normalsize\begin{spacing}{1}
\caption{Contact geometry, uniqueness, and coexistence.}
\label{fig:coexistence-uniqueness}
\end{spacing}}
\end{figure}

\afterfigurespace

At \(c=1/9\), \(A\)'s effort is fixed across equilibria while \(B\)'s lottery varies along this segment: \(L_0\) is pure and every other point semi-pure. The complete type set is therefore \(\{\mathrm P,\mathrm{SP}_A\}\), as permitted by Corollary~\ref{cor:type-coexistence}; two-sided mixing is excluded. Parameterize the segment by \(L_\lambda=\lambda L_1+(1-\lambda)L_0\), so \((p_0,p_1,p_2)=(5\lambda/9,1-\lambda,4\lambda/9)\), \(0\leq\lambda\leq1\). Then \(\Pr(X_B>0)=1-5\lambda/9\), \(\mathbb E[X_B]=\sigma/9\), and \(\operatorname{Var}(X_B)=5\sigma^2\lambda/324\). Increasing \(\lambda\) lowers participation from one to \(4/9\) and raises effort dispersion without changing either contestant's expected effort, cost, prize share, or payoff (Proposition~\ref{prop:outcome-equivalence}).

\begin{samepage}\begin{example}[Seven battlefields: unique finite two-sided mixing]
\label{ex:finite-two-sided}
Let \(J=7\), \(r=c=\Lambda=1\), and \(\widetilde{\boldsymbol q}=(0,\tfrac12,\tfrac12,\tfrac12,\tfrac12,\tfrac12,\tfrac12,1)\), so the prize is split unless one contestant sweeps. The unique equilibrium is \((F,F)\), where \(F=p\delta_a+(1-p)\delta_{sa}\), with \(a\approx0.0349\), \(s\approx3.96\), and \(p\approx0.936\). Online Appendix~\ref{app:finite-two-sided} defines these values exactly and proves uniqueness. The two effort levels balance effort costs against the gains from preventing or completing a sweep.
\end{example}
\end{samepage}

Adjusting the prize schedule by degree elevation (Lemma~\ref{lem:cross-count-identities}) preserves total-effort equilibrium sets, extending Examples~\ref{ex:six-lower-cost-mixing}--\ref{ex:finite-two-sided} to every larger battlefield count. In Table~\ref{tab:equilibrium-by-count}, Examples~\ref{ex:six-lower-cost-mixing} and~\ref{ex:finite-two-sided} establish \(\mathrm{SP}_B\) for \(J\geq6\) and \(\mathrm M_f\) for \(J\geq7\), respectively.\footnote{Theorem~\ref{thm:majority-regimes} below establishes the remaining attainability claims in Table~\ref{tab:equilibrium-by-count}: majority rule attains \(\mathrm P\) for every \(J\), \(\mathrm{SP}_A\) for \(J\geq3\), and \(\mathrm M_\infty\) for \(J\geq7\). At \(r=1\), part~\emph{(iii)} gives a unique equilibrium with infinite positive supports for every \(J\geq7\) and cost ratio.} The seven-battlefield constructions, together with Theorem~\ref{thm:lowcount-refinements}, establish seven as the sharp minimum count for multiplicity or two-sided mixing. The blank entries follow from Theorem~\ref{thm:lowcount-refinements} and the discussion in Section~\ref{S4:counts}.

\section{Infinite-Support Mixing and Uniqueness under Majority Rule}\label{S5}

We first show that infinite-support mixing, whenever it occurs, is independent of the cost ratio. Majority rule nevertheless retains uniqueness beyond the six-battlefield frontier. Figure~\ref{fig:majority-regions} previews the majority regimes; Section~\ref{S5:implications} compares them with KP.

\subsection{A Cost-Independent Infinite-Support Phase}\label{S5:high-elasticity}\label{S5:infinite-phase}

Fix \(J\), \(r\), and an admissible normalized schedule, with reward \(\varphi\). Let \(\mathcal S(\varphi)\) be the set of normalized total-effort distributions \(F\) such that \((F,F)\) is a scalar equilibrium at \(c=1\). Existence, symmetry, and interchangeability make this set nonempty and the equal-cost equilibrium set exactly \(\mathcal S(\varphi)^2\).

\begin{theorem}[Cost-independent infinite-support phase]\label{thm:infinite-phase}
For fixed \(J,r,\widetilde{\boldsymbol q}\), the following are equivalent:
\begin{enumerate}[(i)]
\item An equilibrium has infinite positive support at some cost ratio.
\item An equilibrium has infinite positive support at equal costs.
\item Every equilibrium has infinite positive supports at every cost ratio.
\end{enumerate}
\end{theorem}

When these conditions hold, the complete equilibrium set at each \(c\in(0,1]\) is
\begin{equation}\label{eq:infinite-phase-set}
F_A=F_1,\qquad F_B=(1-c)\delta_0+cF_2,
\qquad F_1,F_2\in\mathcal S(\varphi).
\end{equation}
Here \(\delta_0\) denotes inactivity, and every member of \(\mathcal S(\varphi)\) has no atom at zero, mean \(1/2\), and countably infinite positive support accumulating only at zero.

The proof starts from an infinite-support equilibrium at an arbitrary cost ratio. Let \(\zeta_B\) denote \(B\)'s inactivity probability. Indifference along positive support points approaching zero gives normalized payoffs \((\zeta_B,0)\). Cost balance and deviations that copy the rival's effort lottery then imply \(1-\zeta_B=c\). Keeping \(A\)'s strategy and conditioning \(B\)'s strategy on positive effort yields an equal-cost equilibrium. Appendix~\ref{app:proof-high-elasticity} provides the full proof of the correspondence and its exhaustiveness.

Thus the prize schedule and technology, not the cost ratio, determine whether positive supports are finite or infinite. Within the infinite-support phase, cost asymmetry changes participation: \(A\) always participates and \(B\) participates with probability \(c\), while both conditional effort distributions remain in the equal-cost set. Even without equal-cost uniqueness, lowering \(c\) leaves \(A\)'s expected effort unchanged and reduces \(B\)'s participation. Mean normalized total efforts are \((1/2,c/2)\), normalized prize shares are \((1-c/2,c/2)\), and normalized net payoffs are \((1-c,0)\). In original units, mean efforts are \((\Lambda/2,c\Lambda/2)\), prize shares are \(\bigl(q_0+\Lambda(1-c/2),q_0+\Lambda c/2\bigr)\), each expected effort cost is \(c\Lambda/2\), and net payoffs are \(\bigl(q_0+\Lambda(1-c),q_0\bigr)\).

Starting from an equal-cost equilibrium \((F,F)\) with zero normalized payoffs, keeping \(F\) for \(A\) and using \((1-c)\delta_0+cF\) for \(B\) is familiar from the all-pay robustness literature \citep{alcalde2010rent,ewerhart2017smallnoise}. \Citet{ewerhart2017smallnoise} also conditions on positive effort to derive participation and expected outcomes for arbitrary equilibria. Our result for victory-count schedules makes the correspondence exhaustive throughout the infinite-support phase, without assuming equal-cost uniqueness. Interchangeability allows any pair \(F_1,F_2\in\mathcal S(\varphi)\), while Corollary~\ref{cor:common-support-sets} rules out finite-support alternatives at every cost ratio.

\paragraph{A sufficient elasticity test.}
Reward elasticity, \(\varepsilon_\varphi(t)\equiv t\varphi'(t)/\varphi(t)\), measures the local percentage change in reward per percentage change in the effort ratio. The high-elasticity condition is
\begin{equation}\label{eq:high-elasticity}
\varepsilon_\varphi(t)>1\qquad\text{for every }t\in(0,1].
\end{equation}
\begin{proposition}[High-elasticity test]\label{prop:high-elasticity}
Under~\eqref{eq:high-elasticity}, the equivalent conditions of Theorem~\ref{thm:infinite-phase} hold.
\end{proposition}
In a symmetric equal-cost equilibrium, a smallest positive effort would face only weakly larger rival efforts. High elasticity would then make its marginal reward exceed marginal cost, contradicting optimality. Hence positive supports accumulate at zero. The proposition is a sufficient test for the phase; the general correspondence in Theorem~\ref{thm:infinite-phase} does not require it.

\subsection{Equilibrium Regimes under Majority Rule}\label{S5:majority}

We extend KP's odd-count, equal-cost analysis to arbitrary cost ratios and even counts, and establish uniqueness in every regime. The majority schedule is
\[
q_n^{\mathrm{maj},J}=\mathbb I\{n>J/2\}+\tfrac12\mathbb I\{n=J/2\},
\]
where \(\mathbb I\) equals one when its condition holds and zero otherwise. Splitting a victory-count tie equally is payoff-equivalent to awarding the whole prize by a fair lottery. Pascal's identity gives \(\varphi_{2N+2,r}^{\mathrm{maj}}=\varphi_{2N+1,r}^{\mathrm{maj}}\) for every integer \(N\geq0\). Thus each odd count and the next even count have the same total-effort equilibrium distributions and expected outcomes; only per-battlefield effort differs.

For each pair \(J\in\{2N+1,2N+2\}\), write \(\varphi\equiv\varphi_{J,r}^{\mathrm{maj}}\) and define
\[
D_N\equiv\frac{(2N+1)\binom{2N}{N}}{2^{2N+2}},\qquad
r_J^{\mathrm{all}}=\frac{1}{N+1},\qquad
r_N^{\mathrm{maj}}\equiv\frac{1}{2D_N}.
\]
Here \(D_N\) is each contestant's normalized total effort in the equal-cost pure candidate at \(r=1\). The thresholds \(r_J^{\mathrm{all}}\) and \(r_N^{\mathrm{maj}}\) govern schedule-robust purity and majority-rule high elasticity~\eqref{eq:high-elasticity}, respectively. For \(r_J^{\mathrm{all}}<r\leq\min\{r_N^{\mathrm{maj}},1\}\), let \(\tau_{N,r}\in(0,1]\) be the unique solution of \(\varepsilon_\varphi(\tau_{N,r})=1\).

\begin{theorem}[Unique equilibrium under majority rule]\label{thm:majority-regimes}
Fix an integer \(N\geq0\), \(J\in\{2N+1,2N+2\}\), \(c\in(0,1]\), and \(r\in(0,1]\). Equilibrium is unique.
\begin{enumerate}[(i)]
\item If \(r\leq r_J^{\mathrm{all}}\), equilibrium is pure.
\item If \(r_J^{\mathrm{all}}<r\leq\min\{r_N^{\mathrm{maj}},1\}\), equilibrium is pure when \(c\geq\tau_{N,r}\) and semi-pure when \(c<\tau_{N,r}\).
\item If \(r_N^{\mathrm{maj}}<r\leq1\), equilibrium is two-sided mixed with infinite positive supports.
\end{enumerate}
\end{theorem}

The three regimes partition \((0,1]\). Part~\emph{(i)} follows from FKL's pure-equilibrium bound (Theorem~1 and Proposition~4), recorded in Lemma~\ref{lem:purity-bound}. The intermediate region~\emph{(ii)} is nonempty exactly when \(N\geq1\), and the infinite-support region~\emph{(iii)} exactly when \(N\geq3\), so infinite-support mixing first appears at seven battlefields.\footnote{For \(N\geq1\), \(\binom{2N}{N}<4^N\) gives \(2D_N<(2N+1)/2<N+1\), hence \(r_J^{\mathrm{all}}<r_N^{\mathrm{maj}}\). Moreover, \(r_{N+1}^{\mathrm{maj}}/r_N^{\mathrm{maj}}=2(N+1)/(2N+3)\), so \(r_N^{\mathrm{maj}}\) decreases in \(N\) and is below one exactly when \(N\geq3\).}

\paragraph{Proof of part~(ii).}
Define
\[
C_N\equiv(2N+1)\binom{2N}{N},\qquad
P_N(z)\equiv\sum_{j=0}^{N}\binom{2N+1}{N+1+j}z^j.
\]
The cutoff \(\tau_{N,r}\) is the effort ratio at which \(B\)'s pure-candidate prize just covers her cost. Her payoff relative to inactivity is \(\varphi(c)-c\varphi'(c)=\varphi(c)[1-\varepsilon_\varphi(c)]\). Reward elasticity decreases strictly from \(r(N+1)>1\) near zero to \(2rD_N\leq1\) at one, so this payoff is nonnegative exactly when \(c\geq\tau_{N,r}\).\footnote{The coefficient \(C_N\) counts ways to select a pivotal battlefield and split the other \(2N\) outcomes evenly. Differentiation gives \(\varepsilon_\varphi(t)=rC_N/[(1+t^r)P_N(t^r)]\). The denominator is strictly increasing, and \(P_N(1)=2^{2N}\), so the minimum elasticity on \((0,1]\) is \(2rD_N\). The cutoff equation is equivalently \(rC_N=(1+\tau_{N,r}^{r})P_N(\tau_{N,r}^{r})\).}

The same cutoff determines the supporting-line geometry. Since \(d\log[\varphi(t)/t]/d\log t=\varepsilon_\varphi(t)-1\), the ratio \(\varphi(t)/t\) rises up to \(\tau_{N,r}\) and falls thereafter. Its maximum defines the origin chord of slope \(\sigma_{N,r}\equiv\varphi(\tau_{N,r})/\tau_{N,r}=\varphi'(\tau_{N,r})\). Because the reward is strictly convex and then strictly concave,\footnote{In part~\emph{(ii)}, put \(\eta=r(N+1)>1\). Differentiation gives \(t\varphi''(t)/\varphi'(t)=[(\eta-1)-(\eta+1)t^r]/(1+t^r)\), whose numerator changes sign exactly once, from positive to negative.} the least concave majorant follows the line \(\sigma_{N,r}t\) on \([0,\tau_{N,r}]\) and the reward curve thereafter. Figure~\ref{fig:majority-origin-chord} illustrates the participation cutoff for \(J=3,4\) and \(r=1\), where \(\tau_{N,r}=\sqrt7-2\). Here \(t=X_B/X_A\), with \(X_A>0\). The left panel shows a negative pure-candidate payoff when \(c<\tau_{N,r}\); the right panel shows a positive payoff when \(c>\tau_{N,r}\). The thick chord and its dashed extension represent \(\sigma_{N,r}t\), and the dash-dotted lines are the tangents at \(c\).

\begin{figure}[H]
\centering
\fitwidth{\begin{tikzpicture}[yscale=.78,
  x=1cm,y=1cm,
  axis/.style={-{Latex[length=1.5mm]},black,thin},
  reward/.style={black,line width=.85pt},
  chord/.style={black,line width=1.8pt},
  extension/.style={black,line width=.7pt,dash pattern=on 4pt off 2pt},
  tangent/.style={black,line width=.8pt,dash pattern=on 3.2pt off 1.4pt on .7pt off 1.4pt},
  candidate/.style={black!52,line width=.75pt,dash pattern=on 3.2pt off 1.4pt on .7pt off 1.4pt},
  guide/.style={black!40,line width=.45pt,densely dotted},
  contact/.style={circle,fill=black,draw=black,inner sep=1.6pt},
  hollow/.style={circle,fill=white,draw=black!52,line width=.65pt,inner sep=1.4pt},
  every node/.style={font=\footnotesize,inner sep=1.5pt,text=black}
]
\path[use as bounding box] (-.45,-1.05) rectangle (15.05,4.4);
\pgfmathsetmacro{\majchordtau}{sqrt(7)-2}
\pgfmathsetmacro{\majchordsigma}{(10+7*sqrt(7))/54}
\pgfmathsetmacro{\majchordheight}{\majchordsigma*\majchordtau}

\begin{scope}[x=4.2cm,y=4.5cm]
  \fill[black!9] (0,0) -- (\majchordtau,\majchordheight)
    plot[domain=\majchordtau:0,samples=101]
      (\x,{(\x)^2*(3+\x)/(1+\x)^3}) -- cycle;
  \draw[axis] (0,0) -- (1.53,0);
  \draw[axis] (0,-.09) -- (0,.85);
  \draw[guide] (.5,0) -- (.5,{7/27});
  \draw[guide] (\majchordtau,0) -- (\majchordtau,\majchordheight);
  \draw[reward] plot[domain=0:1.5,samples=151,smooth]
    (\x,{(\x)^2*(3+\x)/(1+\x)^3});
  \draw[chord] (0,0) -- (\majchordtau,\majchordheight);
  \draw[extension] (\majchordtau,\majchordheight)
    -- (1.48,{1.48*\majchordsigma});
  \draw[candidate] (0,{-1/27}) -- (.62,{-1/27+16*.62/27});
  \node[hollow,inner sep=1.1pt] at (0,{-1/27}) {};
  \node[contact,inner sep=1.3pt] at (0,0) {};
  \node[contact] at (\majchordtau,\majchordheight) {};
  \node[hollow] at (.5,{7/27}) {};
  \foreach \u in {.5,\majchordtau}
    \draw[thin] (\u,-.008) -- (\u,.008);
  \node[below] at (.5,-.022) {\(c\)};
  \node[below] at (\majchordtau,-.022) {\(\tau_{N,r}\)};
  \node[above left] at (-.03,.025) {\(0\)};
  \node[right] at (1.56,0) {\(t\)};
  \node[above left] at (1.31,{1.31*\majchordsigma+.025}) {\(\sigma_{N,r}t\)};
  \node[below right] at (1.23,{1.23^2*(3+1.23)/(1+1.23)^3-.025}) {\(\varphi(t)\)};
  \node at (.74,.92) {\(c=1/2<\tau_{N,r}\)};
  \node at (.74,-.175) {\(\varphi(c)-c\varphi'(c)=-1/27<0\)};
\end{scope}

\begin{scope}[xshift=8cm,x=4.2cm,y=4.5cm]
  \draw[axis] (0,0) -- (1.53,0);
  \draw[axis] (0,-.09) -- (0,.85);
  \draw[guide] (1,0) -- (1,.5);
  \draw[reward] plot[domain=0:1.5,samples=151,smooth]
    (\x,{(\x)^2*(3+\x)/(1+\x)^3});
  \draw[tangent] (0,{1/8}) -- (1.5,{1/8+3*1.5/8});
  \node[hollow,draw=black!75,inner sep=1.1pt] at (0,{1/8}) {};
  \node[contact] at (1,.5) {};
  \foreach \u in {\majchordtau,1}
    \draw[thin] (\u,-.008) -- (\u,.008);
  \node[below] at (\majchordtau,-.022) {\(\tau_{N,r}\)};
  \node[below] at (1,-.022) {\(c\)};
  \node[above left] at (-.03,.025) {\(0\)};
  \node[right] at (1.56,0) {\(t\)};
  \node[below right] at (1.23,{1.23^2*(3+1.23)/(1+1.23)^3-.025}) {\(\varphi(t)\)};
  \node at (.74,.92) {\(c=1>\tau_{N,r}\)};
  \node at (.74,-.175) {\(\varphi(c)-c\varphi'(c)=1/8>0\)};
\end{scope}
\end{tikzpicture}}
{\normalsize\begin{spacing}{1}
\caption{Origin-chord and tangent geometry for Theorem~\ref{thm:majority-regimes}\emph{(ii)}.}
\label{fig:majority-origin-chord}
\end{spacing}}
\end{figure}
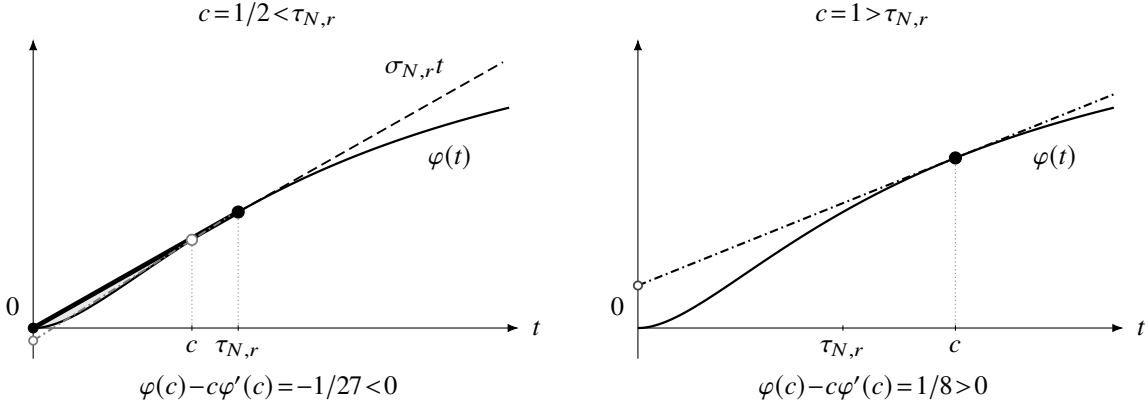

\afterfigurespace

For \(c\leq\tau_{N,r}\), the supporting line \(\sigma_{N,r}t\) has contacts \(\{0,\tau_{N,r}\}\). For \(c<\tau_{N,r}\), the negative tangent intercept at \(c\) excludes the pure candidate. With \(X_A=\sigma_{N,r}\), \(\mathbb E[X_B]=cX_A\) requires \(B\) to choose \(\tau_{N,r}\sigma_{N,r}\) with probability \(c/\tau_{N,r}\) and zero otherwise. At \(c=\tau_{N,r}\), this reduces to the pure candidate.
For \(c>\tau_{N,r}\), the tangent at \(c\) has a single contact, giving normalized total efforts \(X_A^*=\varphi'(c)=rC_Nc^{r(N+1)-1}/(1+c^r)^{2N+2}\) and \(X_B^*=cX_A^*\). These contact conditions make \(B\)'s strategy a best response in both cases.

Since \(B\)'s positive contact is at most one, its reciprocal lies on the majorant's strictly concave part, at or beyond \(\tau_{N,r}\). The tangent there supports \(\varphi\) globally with no other positive contact. Reward symmetry and the test in Section~\ref{S4:uniqueness} make \(A\)'s proposed effort uniquely optimal among positive efforts. If \(B\)'s inactivity probability is \(\zeta_B\), the same inequality at zero gives \(A\) a payoff of at least \(\zeta_B\), compared with \(\zeta_B/2\) from inactivity. The pure and semi-pure candidates are therefore equilibria.

Finally, interchangeability and the exclusion of inactivity for \(A\) (Proposition~\ref{thm:equilibrium-taxonomy}) fix her effort at this unique positive best response in every equilibrium. This remains true at \(\tau_{N,r}=1\), although the reciprocal tangent also touches zero. The contact and mean restrictions then uniquely determine \(B\)'s distribution.

\paragraph{Proof of part~(iii).}
Since reward elasticity decreases strictly on \((0,1]\) and \(\varepsilon_\varphi(1)=2rD_N\), condition~\eqref{eq:high-elasticity} holds exactly when \(r>r_N^{\mathrm{maj}}\). Proposition~\ref{prop:high-elasticity} therefore places the majority-rule contest in the infinite-support phase. By Theorem~\ref{thm:infinite-phase}, uniqueness at every cost ratio then reduces to uniqueness at equal costs.

Lemma~\ref{lem:majority-kernel} in Appendix~\ref{app:majority-uniqueness} establishes equal-cost uniqueness: an integral representation of effort-weighted marginal rewards forces any two equal-cost equilibrium distributions to coincide. Thus \(\mathcal S_{J,r}\equiv\mathcal S(\varphi)=\{F^*_{J,r}\}\), and the correspondence~\eqref{eq:infinite-phase-set} gives the unique equilibrium \(F_A=F^*_{J,r}\), \(F_B=(1-c)\delta_0+cF^*_{J,r}\). The common conditional positive-effort distribution \(F^*_{J,r}\) has no atom at zero, mean normalized effort \(1/2\), and countably infinite positive support accumulating only at zero, completing part~\emph{(iii)}.

\begin{example}[Seven- and eight-battlefield majority]\label{ex:majority-high-elasticity}
By Lemma~\ref{lem:cross-count-identities}, the two counts share the reward function \(\varphi_{7,r}^{\mathrm{maj}}(t)=\varphi_{8,r}^{\mathrm{maj}}(t)=\frac{t^{4r}(35+21t^r+7t^{2r}+t^{3r})}{(1+t^r)^7}\).
Seven and eight are the first counts with \(r_N^{\mathrm{maj}}<1\): \(r_3^{\mathrm{maj}}=32/35\). Thus condition~\eqref{eq:high-elasticity} holds for \(32/35<r\leq1\). Equilibrium is pure for \(r\leq1/4\), pure or semi-pure for \(1/4<r\leq32/35\), and two-sided mixed for \(32/35<r\leq1\). At \(r=1\), \(F^*_{J,1}\) has countably infinite positive support accumulating only at zero; with \(c=1/2\), \(F_A=F^*_{J,1}\) and \(F_B=(\delta_0+F^*_{J,1})/2\), so \(B\) is inactive with probability \(1/2\) and otherwise uses \(F^*_{J,1}\). For \(r=c=1\), Online Appendix~\ref{app:j7-numerical-F} provides a numerical approximation and reports computed deviation gains.
\end{example}

\subsection{Relation to KP}\label{S5:implications}

Specializing to odd \(J\), majority rule, equal costs, and the same Tullock exponent gives KP's simultaneous-election game. KP's per-battlefield campaign effort corresponds to \(X/J\) in our normalized units. KP prove equalization and payoff equivalence, identify the pure candidate and threshold, and establish a symmetric mixed equilibrium above that threshold with no atom at zero and full rent dissipation (zero expected net payoffs). Building on this benchmark, Theorem~\ref{thm:majority-regimes} strengthens payoff equivalence to mixed-equilibrium uniqueness. In this mixed region, the unique equal-cost equilibrium is \((F^*_{J,r},F^*_{J,r})\). Thus uniqueness concerns the entire effort distributions, not only expected expenditures and payoffs.

We also refine the positive-support characterization and its boundary analysis. The distribution \(F^*_{J,r}\) has no atom at zero, mean normalized effort \(1/2\), and countably infinite positive support accumulating only at zero. Effort can therefore be arbitrarily small even though both contestants always participate. Against a fixed positive rival effort, sufficiently small own efforts yield a negative payoff, but this rival-dependent cutoff vanishes as rival effort approaches zero, so it does not bound equilibrium effort away from zero. Our proof treats strategies as probability distributions, including discrete ones. The bound on marginal rewards in Appendix~\ref{app:proof-general-taxonomy} permits averaging derivatives at positive efforts; bounded prizes permit averaging payoff limits as effort approaches zero. This accommodates the discrete support without imposing a density or a positive lower support bound.

Complementing KP's analysis of asymmetric candidates, we identify the equilibrium regime for every finite battlefield count and cost ratio \(c\), including the boundary between pure and semi-pure play. The pure and semi-pure strategies have explicit formulas. In the two-sided mixed region, the theorem identifies a unique conditional positive-effort distribution \(F^*_{J,r}\), its support structure, and its mean; it does not give closed-form atom locations or probabilities. Online Appendix~\ref{app:j7-numerical-F} provides a numerical approximation for seven battlefields at \(r=1\). For any \(c\), \(F^*_{J,r}\) remains the common conditional positive-effort distribution and only \(B\)'s participation probability changes to \(c\). The results also extend to even counts: paired counts \(2N+1\) and \(2N+2\) have identical majority reward functions and normalized total-effort equilibria; only per-battlefield effort changes with \(J\).

\section{Discussion and Conclusion}\label{S6}

\subsection{Beyond the Uniformity Domain}\label{S6:beyond-one}

Our analysis uses \(r\leq1\) to equalize effort across battlefields. Interchangeability does not use \(r\leq1\): the proof of Lemma~\ref{lem:core-reductions}\emph{(ii)} applies verbatim for every \(r>0\). Online Appendix~\ref{OA:beyond-one} proves that existence and common expected outcomes also extend to every finite \(r>0\). The following two examples show what can fail when \(r>1\).

\paragraph{Nonuniformity and multiplicity.}
Realized uniformity and low-count uniqueness need not survive. Take \(J=2\), \(r=2\), \(c=1/2\), and \(\boldsymbol q=(0,1/2,1)\). Contestant \(A\) chooses \((1/4,1/4)\); \(B\) chooses \((1/4,0)\) or \((0,1/4)\) equiprobably. The prize is additive across battlefields, and \(B\)'s coordinate payoff is \(g(z)=\frac{1}{2}\frac{z^2}{z^2+1/16}-z=-\frac{z(z-1/4)^2}{z^2+1/16}\leq0\).
On each battlefield, \(A\)'s positive-effort payoff is \(1/4+g(z)/2\), maximized at \(1/4\); zero yields only \(1/8\). Hence this is an equilibrium. Additivity makes the correlation between \(B\)'s battlefield efforts irrelevant to optimality: equal probabilities on \((0,0)\) and \((1/4,1/4)\) give another equilibrium with the same marginal lotteries. Thus uniform and nonuniform equilibria coexist. More generally, for any \(r>1\), choose \(c^r<r-1\), so that the one-battle contest has no pure equilibrium: the pure candidate gives the higher-cost contestant payoff \(c^r(1+c^r-r)/(1+c^r)^2<0\). Independent draws across battlefields from a one-battle mixed equilibrium (one exists, since for \(J=1\) the existence argument in the proof of Proposition~\ref{thm:existence} does not use \(r\leq1\)), with efforts scaled by \(1/J\), then give nonuniform equilibria under proportional prizes. Thus \(r\leq1\) is the sharp uniformity guarantee across schedules and costs for every \(J\geq2\).

\paragraph{Limits of scalar verification.}
Even an equilibrium of the uniform-effort restriction can fail against nonuniform deviations. Let \(J=3\), \(r=2\), \(c=1\), and \(\boldsymbol q=(0,1/2,1/2,1)\). Here \(\varphi(t)=(3t^2/2+3t^4/2+t^6)/(1+t^2)^3\) satisfies \(\varphi(t)\leq1/8+3t/8\) for \(t\geq0\), with equality only at \(t=1\). Indeed,
\[
\frac18+\frac{3t}{8}-\varphi(t)
=\frac{(t-1)^2\{t^3[3(t-\tfrac16)^2+\tfrac{47}{12}]+5t+1\}}{8(1+t^2)^3}\geq0.
\]
The scalar inequalities in Claim~\ref{prop:general-support} therefore verify total effort \(3/8\) for each contestant in the uniform-effort restriction. Nevertheless, against \((1/8,1/8,1/8)\), the deviation \((1/4,0,0)\) earns \((1/2)(4/5)-1/4=3/20>1/8\). Thus scalar optimality does not rule out profitable reallocations across battlefields.

\subsection{Concluding Remarks}\label{S6:conclusion}

The central results concern uniqueness and support beyond a schedule-robust pure-equilibrium region. Equilibrium is unique through six battlefields for every admissible schedule and cost ratio, while seven first permits multiplicity or two-sided mixing. For arbitrary schedules, the occurrence of infinite positive supports is independent of the cost ratio; throughout that phase, the entire asymmetric-cost equilibrium set follows from the equal-cost set by adjusting participation. Majority rule retains uniqueness at every battlefield count and cost ratio. We identify its pure, semi-pure, and infinite-support regions. Across general schedules, multiplicity can change participation and dispersion while preserving expected efforts, prize shares, costs, and payoffs. Existence and uniformity make these results apply to unrestricted vector strategies throughout \(0<r\leq1\).

For \(0<r\leq1\), the analysis extends to total-effort power costs \(C_i(\mathbf e_i)=d_iE_i^\gamma\), \(\gamma\geq1\), with finite expected costs. Equalization at fixed \(E_i\) is unchanged. Relabel so that \(0<d_A\leq d_B\) and set \(Z_i=d_BE_i^\gamma/\Lambda\). For uniform allocations, total-effort ratios equal \((Z_A/Z_B)^{1/\gamma}\), so the transformed contest has exponent \(r/\gamma\), linear costs, and cost ratio \(d_A/d_B\); the existence, equalization, taxonomy, and verification results therefore apply. Outcome invariance concerns \(\mathbb E[E_i^\gamma]\), expected costs, prize shares, and payoffs; for \(\gamma>1\), it need not extend to \(\mathbb E[E_i]\). Theorems~\ref{thm:lowcount-refinements}, \ref{thm:infinite-phase}, and~\ref{thm:majority-regimes} also apply, with \(r\) replaced by \(r/\gamma\) in hypotheses and conclusions and efforts measured by \(Z_i\); the examples, which fix \(r=1\), correspond to \(\gamma=1\).

Extending the mixed-equilibrium characterization to schedules that depend on which battlefields are won, not only how many, is a natural next step. A further question is how an organizer should choose the prize schedule to maximize total effort beyond FKL's schedule-robust purity bound when participation and strategic randomization are endogenous.\par

\appendix
\numberappendixobjectsbysection

\section{Proofs for Section~\ref{S2}}\label{app:paper-core}

\subsection{Proofs for Section~\ref{S2:reduction}}
\label{app:core-reductions}
\label{app:existence-scalar-reduction}

\begin{proof}[Proof of Lemma~\ref{lem:core-reductions}]
Let \(\mathbf 1_J\) be the \(J\)-vector of ones and write \(\mathcal V\equiv\{(E/J)\mathbf 1_J:E\geq0\}\) and \(\mathcal V_+\equiv\mathcal V\setminus\{\mathbf 0\}\) for the sets of uniform and positive uniform effort vectors. For contestant \(i\in\{A,B\}\), let \(-i\) denote her opponent and \(F_i\) her mixed strategy. Because \(\mathcal V\) is closed, \(F_i(\mathcal V)=1\) is equivalent to \(\operatorname{supp}F_i\subseteq\mathcal V\).

Part~\emph{(i)} is immediate for \(J=1\). For \(J\geq2\), we reproduce FKL's general-schedule equalization argument (Proposition~1 and Online Appendix~A.1), which extends KP's pairwise argument (Appendix Lemma~1). Set \(\Delta_n\equiv q_n-q_{n-1}\), and fix positive uniform rival effort \(a\). Monotonicity gives \(\Delta_n\geq0\). Let \(G_a(\mathbf e_i)\) denote contestant \(i\)'s expected gross reward from effort vector \(\mathbf e_i\), and write \(p(x)=x^r/(x^r+a^r)\). Hold all own efforts except \(x\) and \(y\) fixed. If \(W'\) is the number of victories on the remaining \(J-2\) battlefields, put \(\delta_1=\mathbb E[\Delta_{W'+1}]\) and \(\delta_2=\mathbb E[\Delta_{W'+2}]\). The reward contribution of \(x\) and \(y\) is
\[
\begin{aligned}
R_{\mathrm{pair}}(x,y)&=\delta_1\{p(x)[1-p(y)]+[1-p(x)]p(y)\}+(\delta_1+\delta_2)p(x)p(y),\\
(R_{\mathrm{pair}})_x-(R_{\mathrm{pair}})_y&=\delta_1\{p'(x)[1-p(y)]-p'(y)[1-p(x)]\}
+\delta_2\{p'(x)p(y)-p'(y)p(x)\}.
\end{aligned}
\]
Both \(p'(x)/p(x)=ra^r/[x(x^r+a^r)]\) and \(p'(x)/[1-p(x)]=rx^{r-1}/(x^r+a^r)\) are strictly decreasing on \((0,\infty)\) when \(r\leq1\). Hence \(x>y>0\) makes both brackets in \((R_{\mathrm{pair}})_x-(R_{\mathrm{pair}})_y\) nonpositive, so \((x-y)((R_{\mathrm{pair}})_x-(R_{\mathrm{pair}})_y)\leq0\). When every effort component is positive, every value \(0,\ldots,J-2\) of \(W'\) occurs with positive probability. Since \(\sum_{n=1}^{J}\Delta_n=q_J-q_0>0\), this implies \(\delta_1+\delta_2>0\); the inequality is therefore strict when \(x\ne y\). Transferring effort from larger to smaller components strictly raises gross reward until they coincide, preserving total effort. Transfers from above-mean to below-mean components, stopping when one reaches the mean, yield the uniform vector in finitely many strictly improving steps. This proves uniqueness among positive vectors with a fixed total.

To include zero components, let \(\mathbf e_i\) be nonuniform with total \(E_i>0\), set \(\mathbf u_i=(E_i/J)\mathbf 1_J\), and define \(\mathbf e_i^\varepsilon=(1-\varepsilon)\mathbf e_i+\varepsilon\mathbf u_i\) for \(0<\varepsilon<1\). The weak pairwise comparison extends to zero components by continuity. Transfers from above-mean to below-mean components, stopping at their targets in \(\mathbf e_i^\varepsilon\), therefore give \(G_a(\mathbf e_i)\leq G_a(\mathbf e_i^\varepsilon)\). Since \(\mathbf e_i^\varepsilon\) is positive and nonuniform, the preceding strict comparison gives \(G_a(\mathbf e_i^\varepsilon)<G_a(\mathbf u_i)\). Averaging over positive uniform rival efforts preserves this strict gain; against an inactive rival, uniform allocation wins every battlefield and is weakly better. Equalizing each realized vector preserves expected cost and strictly improves any mixture that assigns positive probability to nonuniform vectors when the rival participates with positive probability.

Part~\emph{(ii)} is also established in FKL's Lemma~A.1. We adapt the interchangeability argument in KP (p.~1104) and \citet[Appendix Lemma~A.1]{ewerhart2017revenue} to normalized asymmetric costs. Let\nobreak\hspace{.25em}\(\mathcal G_i\) denote contestant \(i\)'s expected normalized gross reward and\nobreak\hspace{.25em}\(\chi_i\) her expected normalized effort cost. Budget balance gives \(\mathcal G_A+\mathcal G_B=1\). Starting from \(\widehat\pi_i=\mathcal G_i-\chi_i\), define \(\widetilde\pi_i=\widehat\pi_i+\chi_{-i}\).
Each added term depends only on the opponent's strategy, so best responses are unchanged, while \(\widetilde\pi_A+\widetilde\pi_B=1\). Let \(v_F=\widetilde\pi_A(F_A,F_B)\) and \(v_G=\widetilde\pi_A(G_A,G_B)\) for two equilibria \(F\) and \(G\). Their equilibrium best-response inequalities give \(v_F\leq \widetilde\pi_A(F_A,G_B)\leq v_G\leq \widetilde\pi_A(G_A,F_B)\leq v_F\). Thus equality holds throughout, and both cross-pairs are equilibria. Linearity in mixing probabilities also makes each equilibrium-strategy set convex.

For part~\emph{(iii)}, the cross-equilibrium support step parallels \citet[Lemma~2]{ewerhartsun2018equilibrium}. Fix an arbitrary equilibrium \(F\) and contestant \(i\). Part~\emph{(ii)} pairs \(F_i\) with \(\bar F_{-i}\) from the equilibrium in the premise. Since \(\bar F_{-i}(\mathcal V)=1\) and \(\bar F_{-i}(\mathcal V_+)>0\), part~\emph{(i)} gives \(F_i(\mathcal V)=1\). This holds for both contestants, proving the claim.
\end{proof}

\begin{proof}[Proof of Proposition~\ref{thm:existence}]
\emph{Existence.} Consider the uniform-effort restriction. Let \(x\) and \(y\) denote the normalized total efforts of \(A\) and \(B\). Write \(\varphi=\varphi_{J,r}(\,\cdot\,;\widetilde{\boldsymbol q})\), and set \(K_c(x,y)=\mathcal U(x,y)-cx+y\); the transformed payoffs are \(K_c\) and \(1-K_c\). For mixed strategies, \(K_c(x,F_B)\) and \(K_c(F_A,y)\) denote the corresponding expected payoffs. Since \(0\leq\mathcal U\leq1\), efforts above \(1/c\) for \(A\) and one for \(B\) are dominated by inactivity. Restrict actions accordingly. For \(\varepsilon>0\), replace \(\mathcal U\) by \(\mathcal U_\varepsilon(x,y)=\varphi((x+\varepsilon)/(y+\varepsilon))\), which removes the discontinuity at \((0,0)\), and denote the resulting transformed payoff by \(K_c^\varepsilon\). Continuity on these compact intervals and the constant-sum minimax theorem give a mixed equilibrium \((F_A^\varepsilon,F_B^\varepsilon)\) and a transformed equilibrium payoff \(v_\varepsilon\) for \(A\). The same domination bounds exclude deviations outside these intervals.
Uniformly bounded equilibrium values and compactness of distributions on these intervals give a sequence \(\varepsilon\downarrow0\) along which \(v_\varepsilon\to v\) and \((F_A^\varepsilon,F_B^\varepsilon)\) converges weakly to \((F_A,F_B)\): expectations of bounded continuous functions converge.

For fixed \(x>0\), \(K_c^\varepsilon(x,\cdot)\) converges uniformly on \([0,1]\) to the continuous function \(K_c(x,\cdot)\). Together with weak convergence, this gives \(K_c^\varepsilon(x,F_B^\varepsilon)\to K_c(x,F_B)\). Similarly, for each fixed \(y>0\), \(K_c^\varepsilon(F_A^\varepsilon,y)\to K_c(F_A,y)\). Passing to the limit in the equilibrium inequalities yields \(K_c(x,F_B)\leq v\leq K_c(F_A,y)\) for all \(x,y>0\).

For zero-effort deviations, let \(\zeta_i=F_i(\{0\})\). As positive effort tends to zero, the normalized expected prize remains one against an inactive rival and tends to zero against any fixed positive rival effort. Bounded prizes allow these limits to be averaged. Zero effort instead yields one half against an inactive rival. Hence
\[
\begin{aligned}
K_c(0,F_B)&=\mathbb E[X_B]+\zeta_B/2
\leq\lim_{X_A\downarrow0}K_c(X_A,F_B)=\mathbb E[X_B]+\zeta_B\leq v,\\
K_c(F_A,0)&=1-c\mathbb E[X_A]-\zeta_A/2
\geq\lim_{X_B\downarrow0}K_c(F_A,X_B)=1-c\mathbb E[X_A]-\zeta_A\geq v.
\end{aligned}
\]
The inequalities hold for all \(x,y\geq0\); integrating over \(F_A,F_B\) gives \(K_c(F_A,F_B)\leq v\leq K_c(F_A,F_B)\), hence equality. Averaging over admissible alternatives rules out randomized deviations; expectations are finite because \(|K_c(x,y)|\leq1+cx+y\) and admissible strategies have finite expected effort. Thus \((F_A,F_B)\) is a total-effort equilibrium.

\emph{Uniformity and equilibrium correspondence.} Against an always-inactive rival, normalized payoff has supremum one, approached by positive efforts tending to zero but unattained by any pure or mixed strategy. Thus every scalar equilibrium has positive participation by both contestants. The equalization argument in Lemma~\ref{lem:core-reductions}\emph{(i)} compares each realized effort vector with its uniform counterpart at the same total effort. Nonuniform deviations therefore cannot improve on scalar deviations, so every scalar equilibrium is also an equilibrium of the original contest.

The constructed equilibrium is uniform; Lemma~\ref{lem:core-reductions}\emph{(iii)} makes every original equilibrium uniform. Hence the scalar restriction preserves the equilibrium set.
\end{proof}

\subsection{Proofs for Section~\ref{S2:equilibrium-forms}}\label{app:proof-general-taxonomy}

The following lemma bounds the number of zeros of a finite sum of real powers. A zero has \emph{multiplicity} \(m\) if the function's first nonzero derivative there is the \(m\)th derivative.
\begin{samepage}
\begin{lemma}\label{lem:generalized-polynomial-zeros}
Let \(N\geq1\), let \(\lambda_1<\cdots<\lambda_N\) be real exponents, and let \(a_1,\ldots,a_N\in\mathbb R\). If \(g(t)=\sum_{k=1}^{N}a_kt^{\lambda_k}\) is not identically zero, then \(g\) has at most \(N-1\) positive zeros, counted with multiplicity.
\end{lemma}
\end{samepage}

\begin{proof}

Real powers, hence \(g\), are analytic for \(t>0\); since \(g\) is not identically zero, each zero has finite multiplicity. Discard zero coefficients and relabel; induct on their number \(N\). For \(N=1\), \(g(t)=a_1t^{\lambda_1}\) has no positive zeros.

For the inductive step, let \(N\geq2\) and define \(h(t)=t^{-\lambda_1}g(t)=a_1+\sum_{k=2}^{N}a_kt^{\lambda_k-\lambda_1}\). Multiplication by \(t^{-\lambda_1}>0\) preserves positive zeros and their multiplicities. Differentiating removes the constant term and gives \(h'(t)=\sum_{k=2}^{N}a_k(\lambda_k-\lambda_1)t^{\lambda_k-\lambda_1-1}\). Its smallest-power coefficient is \(a_2(\lambda_2-\lambda_1)\ne0\); dividing by the corresponding power of \(t\) and taking \(t\downarrow0\) shows that \(h'\) is not identically zero. This is a sum of \(N-1\) nonzero terms with distinct exponents, so the induction hypothesis bounds its positive zeros by \(N-2\), counted with multiplicity.

For any finite collection of positive zeros \(t_1<\cdots<t_s\) of \(h\), let their multiplicities be \(m_j\) and set \(M=\sum_jm_j\). The derivative has multiplicity \(m_j-1\) at \(t_j\); Rolle's theorem adds a zero between consecutive points. Thus \(h'\) has at least \(M-1\) positive zeros, giving \(M-1\leq N-2\) and \(M\leq N-1\). Applied to every finite collection, this also rules out infinitely many zeros.
\end{proof}

\begin{proof}[Proof of Proposition~\ref{thm:equilibrium-taxonomy}]
By Proposition~\ref{thm:existence}, we work with normalized total efforts. Fix an equilibrium with distributions \(F_A,F_B\), and write \(\zeta_i=F_i(\{0\})\).

 \Needspace{6\baselineskip}
 \noindent\emph{Part (i): Equilibrium and participation.}\par\nopagebreak
\paragraph{Step 1: Participation and cost balance.}
Neither contestant is always inactive, as shown in the existence proof.
Nor can both contestants have an atom at zero: if inactivity is a best response while the opponent is inactive with positive probability, replacing zero by a sufficiently small positive effort raises expected gross reward by at least one half of that probability, weakly raises it against every positive rival effort, and has arbitrarily small cost. Hence \(\zeta_A\zeta_B=0\).

For \(x,y>0\), define \(m(x,y)\equiv x\mathcal U_x(x,y)=-y\mathcal U_y(x,y)=rp(1-p)Q_J'(p)\), where \(p=x^r/(x^r+y^r)\) and subscripts denote partial derivatives. Since the polynomial \(Q_J'\) is bounded on \([0,1]\) and \(p(1-p)\leq1/4\), \(|m(x,y)|\leq C_m\) for some finite \(C_m\). When exactly one effort is zero, \(p\) is zero or one, so we extend this formula by setting \(m=0\). Set \(m(0,0)=0\) as well; since \(\zeta_A\zeta_B=0\), this value has no effect on equilibrium averages. For \(x\geq a>0\), \(|\mathcal U_x(x,y)|\leq C_m/a\) for every rival effort \(y\). This uniform bound justifies differentiating under the expectation. The same argument applies to \(B\).
Equilibrium assigns probability one to best responses. Continuity of expected payoff at positive efforts then makes every positive support point a best response: otherwise a neighborhood of suboptimal efforts would receive positive probability.

At positive efforts \(x\) and \(y\) in the respective equilibrium supports, expected marginal prize equals marginal cost. For \(B\), the normalized prize is \(1-\mathcal U(x,y)\). Suppressing arguments, multiplying each first-order condition by own effort gives
\[
\int \mathcal U_x\,dF_B=c,\quad -\int \mathcal U_y\,dF_A=1
\quad\Longrightarrow\quad
\int m\,dF_B=cx,\quad \int m\,dF_A=y.
\]
The first identity on the right also holds at \(x=0\), and the second at \(y=0\): both sides are then zero by the definition of \(m\). No first-order condition at zero is needed.
Averaging these right-hand identities over \(F_A\) and \(F_B\), respectively, gives the same average of \(m\); boundedness of \(m\) permits exchanging the order of integration. Hence
\begin{equation}\label{eq:app-general-cost-equality}
c\mathbb E[X_A]
=\iint m(x,y)\,dF_A(x)dF_B(y)
=\mathbb E[X_B]
\equiv\chi.
\end{equation}

\paragraph{Step 2: Which contestant may be inactive?}
Suppose the lower-cost contestant \(A\) had \(\zeta_A>0\). Then \(\zeta_B=0\), and inactivity gives \(A\) an equilibrium payoff of zero. If \(A\) copies \(B\)'s strategy by drawing independently from \(F_B\), symmetry and budget balance give her an expected prize of \(1/2\); \(A\)'s expected cost is \(c\mathbb E[X_B]=c\chi\). Optimality would require \(\chi\geq1/(2c)\geq1/2\). But \(B\)'s payoff equals the unit prize minus both expected costs, \(1-2\chi\), because \(A\) earns zero. Inactivity gives her at least \(\zeta_A/2>0\), so \(\chi<1/2\). This contradiction proves \(\zeta_A=0\). If \(c=1\), relabeling the contestants also gives \(\zeta_B=0\). Thus only the strictly higher-cost contestant can have an atom at zero.

 Since neither contestant is always inactive (Step~1), any fixed total effort must be positive. The pure, semi-pure, and two-sided cases are therefore exhaustive, completing part~(i).

\noindent\emph{Part (ii): Support restrictions.}\par\nopagebreak
\paragraph{Step 3: Positive support is discrete.}
Fix either contestant and the opponent's equilibrium distribution. By Step~1, every positive support point yields the same equilibrium payoff. Expected payoff is real analytic on \((0,\infty)\).{\footnote{Fix \(x_0>0\), put \(u=(x/x_0)^r-1\), and set \(a_y=x_0^r/(x_0^r+y^r)\in(0,1)\) for \(y>0\). The win probability is \(a_y(1+u)/(1+a_yu)\). Reward is polynomial of degree at most \(J\) in this probability, hence has denominator \((1+a_yu)^J\) and a polynomial numerator with coefficients uniformly bounded in \(y\). For \(|u|\leq1/2\), the reciprocal denominator's binomial series has absolute sum \((1-|a_yu|)^{-J}\leq2^J\). These bounds permit termwise expectation and yield a convergent power series for expected reward. Since \(u\) is analytic in \(x\) near \(x_0\), expected reward is analytic; subtracting linear cost preserves analyticity. An inactive rival adds only a constant for \(x>0\).}} Subtracting the equilibrium payoff gives an analytic function that vanishes at every positive support point. If these zeros accumulated above zero, the identity theorem would make this function identically zero. But bounded prizes and positive marginal cost imply that payoff tends to \(-\infty\) as effort grows, a contradiction.

The domination bounds from the existence proof confine the supports to \([0,1/c]\) for \(A\) and \([0,1]\) for \(B\). Thus only finitely many support points lie above each \(\varepsilon>0\); otherwise compactness would produce a positive accumulation point. Taking \(\varepsilon=1/n\) expresses the positive support as a countable union of finite sets: it is finite or countably infinite, with accumulation possible only at zero. Each positive support point is isolated and has positive probability; zero need not be an atom.

\paragraph{Step 4: Linking the two supports.}
Suppose the opponent has finite positive support \(y_1,\ldots,y_L\). With \(\widetilde\Delta_{n+1}=\widetilde q_{n+1}-\widetilde q_n\), the effort-weighted marginal reward \(m\) from Step~1 is
\[
m(x,y_\ell)=rJ\sum_{n=0}^{J-1}\widetilde\Delta_{n+1}\binom{J-1}{n}
\frac{x^{r(n+1)}y_\ell^{r(J-n)}}{(x^r+y_\ell^r)^{J+1}}.
\]
For either contestant \(i\), write \(w_\ell=F_{-i}(\{y_\ell\})>0\), and let \(d=c\) for \(A\) and \(d=1\) for \(B\). Prize symmetry makes the displayed expression valid with \(x\) denoting either contestant's own effort. The weighted first-order condition is \(R_d(x)\equiv\sum_{\ell=1}^{L}w_\ell m(x,y_\ell)-dx=0\); a possible rival atom at zero contributes nothing.

Multiplying \(R_d(x)=0\) by \(\prod_{\ell=1}^{L}(x^r+y_\ell^r)^{J+1}\) yields a finite sum of powers \(x^\lambda\). The resulting expression is not identically zero: the multiplier is positive, and \(R_d(x)\) is \(x\) times the derivative of expected payoff. If it vanished identically, payoff would be constant on \((0,\infty)\), contradicting bounded reward and unbounded effort cost. Lemma~\ref{lem:generalized-polynomial-zeros} therefore gives only finitely many positive roots.

Thus a semi-pure mixer has finite support. In a two-sided mixed equilibrium, the positive supports are either both finite or both countably infinite, accumulating only at zero. This proves part~(ii).
\end{proof}

\section{Proofs and Computations for Section \ref{S4}}\label{app:application-proofs}

\subsection{Proofs for Section~\ref{S4:counts}}\label{app:proof-battlefield-count}

The equilibrium guarantee and its sharpness in the following lemma are due to FKL (Theorem~1 and Propositions~3--4). We give a self-contained scalar proof.

 \begin{lemma}[Sharp schedule-robust purity bound]\label{lem:purity-bound}
Put \(L=\lceil J/2\rceil\). If \(0<r\leq1/L\), then \(\varphi_{J,r}\) is strictly concave on \((0,\infty)\), and equilibrium is unique and pure for every admissible schedule and cost ratio. For every \(J\geq3\) and \(r>1/L\) in the maintained domain, majority rule has no pure equilibrium for all sufficiently small \(c\).
\end{lemma}

\begin{proof}
For strict concavity, write \(\varphi=\varphi_{J,r}\), \(K=J+1\), and \(z=t^r\), and set
\[
a_k=\binom{J-1}{k-1}(\widetilde q_k-\widetilde q_{k-1}),\qquad
M(z)=\frac{\sum_{k=1}^{J}a_kz^k}{(1+z)^K}.
\]
Then \(t\varphi'(t)=rJM(z)>0\). Monotonicity and symmetry make \(M\) a nonzero nonnegative combination of \(h_k(z)=(z^k+z^{K-k})/(1+z)^K\), \(1\leq k\leq L\), with central coefficient \(a_k/2\) when \(K=2k\). For \(d=K-2k\geq0\),
\[
\frac{z h_k'(z)}{h_k(z)}
=k+\frac{d z^d}{1+z^d}-\frac{Kz}{1+z}.
\]
For \(0<z\leq1\), this is at most \(k-2kz/(1+z)<k\leq L\), using \(z^d/(1+z^d)\leq z/(1+z)\) when \(d\geq1\); the \(d=0\) case is immediate. For \(z>1\), it is strictly below \(K/2-k<L\). Thus \(zM'(z)/M(z)\), a weighted average of these ratios, is below \(L\), giving
\[
\frac{t\varphi''(t)}{\varphi'(t)}
=r\frac{zM'(z)}{M(z)}-1<rL-1\leq0.
\]
Strict concavity makes both tangent inequalities in Claim~\ref{prop:general-support} strict at every alternative ratio, including zero. Thus the candidate is a pure equilibrium, and interchangeability gives uniqueness.

Under majority rule, the first rewarded victory count is \(L\). Expanding~\eqref{eq:normalized-payoff-kernel} at zero gives
\[
\varphi(t)-t\varphi'(t)=C(1-rL)t^{rL}+O(t^{r(L+1)}),
\qquad C=\binom{J}{L}q_L^{\mathrm{maj},J}>0.
\]
When \(rL>1\), contestant~\(B\)'s pure-candidate payoff is negative for all sufficiently small \(c\), so inactivity is profitable.
\end{proof}

 \begin{lemma}[One-chord geometry through six battlefields]
\label{lem:low-count-chord-geometry}
If \(J\in\{3,4,5\}\), the least concave majorant of \(\varphi_{J,r}\) differs from the reward curve on at most one open interval, whose endpoints satisfy \(0\leq t_-<t_+<1\). If \(J=6\), there is again at most one such interval, and its two endpoints lie either both below one or both above one. In either case, on the side of \(t=1\) containing the reciprocals of the positive chord endpoints, the reward curve equals its least concave majorant and is strictly concave; hence each reciprocal-deviation inequality is strict away from its reciprocal contact ratio.
\end{lemma}

\begin{proof}
We first restrict the curvature changes, then use them to locate supporting chords and verify the reciprocal tangents.

\emph{A root-count bound.} Descartes' rule in the Bernstein basis bounds the number of roots of a nonzero real degree-\(d\) polynomial \(p(z)=\sum_{k=0}^{d}b_k\binom{d}{k}z^k(1-z)^{d-k}\) in \((0,1)\), counted with multiplicity, by the number of sign changes in \((b_0,\ldots,b_d)\) after zero coefficients are deleted.

\emph{Three through five battlefields.}
In this proof, write \(\Phi=\varphi_{J,r}\) and \(g=\varphi_{J,1}\) for the fixed normalized schedule, so \(\Phi(t)=g(t^r)\), and set \(z=t^r\). The chain rule gives \(\Phi''(t)=rt^{r-2}\bigl[(r-1)g'(z)+rzg''(z)\bigr]\).
For \(J=3\), write the normalized schedule as \((0,\alpha,1-\alpha,1)\), with \(0\leq\alpha\leq1/2\). For a normalized four-battlefield schedule \((0,a,1/2,1-a,1)\), direct expansion of~\eqref{eq:normalized-payoff-kernel} gives
\[
g(z)=\frac{4az+3z^2+4(1-a)z^3+z^4}{(1+z)^4}
=\frac{3\alpha z+3(1-\alpha)z^2+z^3}{(1+z)^3},\qquad \alpha=\frac{4a}{3}\in[0,2/3].
\]
Both counts therefore use the same reward function \(g_\alpha\).
The numerator governing \(\Phi''\) has degree-three Bernstein coefficients \(\alpha(r-1)\), \(\frac{-2+4r-8r\alpha}{3}\), \(\frac{4r-6+\alpha(6-8r)}{3}\), and \(-4+4\alpha\).
The first, third, and fourth are nonpositive, with the last two strictly negative. Hence the sequence has at most two sign changes, so \(\Phi''\) has at most two zeros on \((0,1)\).

\emph{The five-battlefield curvature bound.} For \(J=5\), put \(a=\widetilde q_1\) and \(b=\widetilde q_2\), so \(0\leq a\leq b\leq1/2\), and define \(P(z)=az^4+4(b-a)z^3+6(1-2b)z^2+4(b-a)z+a\).
Direct differentiation gives \(g'(z)=5P(z)/(1+z)^6\), with \(P(z)>0\) for \(z>0\). Removing positive factors from \(\Phi''\) gives the curvature numerator
\(K_r(z)=(r-1)(1+z)P(z)+rz[(1+z)P'(z)-6P(z)]\).
Its six Bernstein coefficients are the convex combination
\((1-2b)V_0+2(b-a)V_1+2aV_2\), where the three weights are nonnegative and sum to one:
\[
\begin{aligned}
V_0&=\left(0,0,\frac{3(3r-1)}{5},\frac{6(3r-2)}{5},\frac{6(3r-5)}{5},-12\right),\\
V_1&=\left(0,\frac{2(2r-1)}{5},\frac{4r-5}{5},\frac{4r-10}{5},\frac{4r-20}{5},-8\right),\\
V_2&=\left(\frac{r-1}{2},-\frac{3}{5},-\frac{7+5r}{10},-\frac{4+5r}{5},-(1+r),-2\right).
\end{aligned}
\]
Thus the first coefficient is nonpositive, the last two are negative, and the third is at least the fourth. More than two sign changes would require the third coefficient to be negative and the fourth positive, which is impossible. Hence \(\Phi''\) again has at most two zeros on \((0,1)\).

\emph{A strict supporting tangent at one.} For all three counts, \(g'(z)>0\), \(g''(z)<0\) on \([1,\infty)\), and the tangent to \(g\) at one lies strictly above \(g(z)\) for \(z\neq1\). For \(J=3,4\), \(g_\alpha'>0\) because \(g_\alpha\) is the reward of the admissible schedule \((0,3\alpha/4,1/2,1-3\alpha/4,1)\). Writing \(z=1+v\), \((1+z)^5g_\alpha''(z)=3[-(4-4\alpha)+(10\alpha-6)v-2\alpha v^2]\), which is negative for \(v\geq0\): all coefficients are nonpositive when \(\alpha\leq3/5\), and the discriminant \(132\alpha^2-152\alpha+36\) is negative for \(3/5<\alpha\leq2/3\). Direct simplification also gives
\[
8(1+z)^3\{g_\alpha(1)+g_\alpha'(1)(z-1)-g_\alpha(z)\}
=(z-1)^2[1+3\alpha+(8-12\alpha)z+(3-3\alpha)z^2]>0.
\]
For \(J=5\), \(-g''\) and the tangent gap are affine in \((a,b)\). The parameter triangle \(0\leq a\leq b\leq1/2\) has vertices \((0,0),(0,1/2),(1/2,1/2)\), so positivity at its vertices implies positivity throughout. The vertex numerators of \(-g''\), up to positive factors, are respectively
\(6z(2z-1)\), \(3z^3-3z^2+5z-1\), and \(z^4/2-z^3+3/2\), all positive for \(z\geq1\); the corresponding tangent-gap numerators after removing \((z-1)^2\) are
\[
30z^4+148z^3+256z^2+44z+2,   \quad
 20z^4+88z^3+96z^2+104z+12,   \quad
 5z^4-2z^3+16z^2+34z+27.
\]
The first two have only positive coefficients; the third is positive for \(z\geq0\), since \(5z^4-2z^3+16z^2+34z+27=z^2[5(z-\frac{1}{5})^2+\frac{79}{5}]+34z+27>0\).
 Because \(g\) is increasing and strictly concave on \([1,\infty)\), while \(t\mapsto t^r\) is increasing and concave, \(\Phi(t)=g(t^r)\) is strictly concave there. Moreover, for \(t\ne1\), \(\Phi(t)=g(t^r)<g(1)+g'(1)(t^r-1)\leq\Phi(1)+\Phi'(1)(t-1)\), where the last inequality is \(t^r-1\leq r(t-1)\). Thus the tangent to \(\Phi\) at one is a strict global upper bound.

\emph{Constructing and locating the chord.} Curvature changes are sign changes of \(\Phi''\); concave or convex pieces are maximal such intervals, unsplit by zeros without a sign change. The root count on \((0,1)\) and strict concavity on \([1,\infty)\) leave only three patterns: concave; convex then concave; or concave, then convex, then concave. In the last case, let \(I_-\) and \(I_+\) be the closures of the first and last concave pieces; in the convex--concave case, take \(I_-=\{0\}\) and let \(I_+\) be the closure of the final concave piece. Define \(V_\pm(\sigma)=\max_{t\in I_\pm}\{\Phi(t)-\sigma t\}\).
 Each \(V_\pm(\sigma)\) is the intercept of the lowest line of slope \(\sigma\) above its piece. For every \(\sigma>0\), the maxima exist because \(\Phi\) is bounded and \(-\sigma t\) tends to minus infinity; the value functions are continuous in \(\sigma\). Let \(\Delta_V(\sigma)=V_+(\sigma)-V_-(\sigma)\). If \(\sigma_2>\sigma_1\), choose maximizers \(t_+^2\in I_+\) at \(\sigma_2\) and \(t_-^1\in I_-\) at \(\sigma_1\). Optimality gives \(\Delta_V(\sigma_2)-\Delta_V(\sigma_1)\leq-(\sigma_2-\sigma_1)(t_+^2-t_-^1)<0\), because \(I_+\) lies strictly to the right of \(I_-\). Thus \(\Delta_V\) is strictly decreasing. It is positive for sufficiently small \(\sigma\), because \(V_+(\sigma)\to1\) whereas \(\max_{t\in I_-}\Phi(t)<1\), and negative for sufficiently large \(\sigma\), because \(I_+\) is bounded away from zero. Hence exactly one slope \(\sigma^*\) gives equal intercepts. The corresponding line lies above each concave piece by construction and above the intervening convex piece because the curve minus the line is convex and nonpositive at both endpoints. Replacing \(\Phi\) by this line between its contacts gives the unique supporting chord, with the least concave majorant strictly above \(\Phi\) in its interior. At \(\sigma=\Phi'(1)\), the strict global tangent bound at one gives \(\Delta_V(\sigma)>0\), so \(\sigma^*>\Phi'(1)\). Because \(\Phi'\) decreases on the final concave piece, its upper contact satisfies \(\Phi'(t_+)=\sigma^*>\Phi'(1)\), and hence \(t_+<1\).

\emph{Six battlefields.}
It is enough to show that all curvature changes occur on at most one side of the unit ratio, with at most two changes on that side.
Write the normalized schedule as \((0,a,b,1/2,1-b,1-a,1)\), where \(0\leq a\leq b\leq1/2\), and again write \(\Phi(t)=g(t^r)\). Let \(z=t^r\), \(w=(z-1)/(z+1)\), and \(s=w^2\).
The map \(w=(z-1)/(z+1)\) sends \(z>0\) to \((-1,1)\); reciprocal ratios change \(w\) to \(-w\), leaving \(s=w^2\) unchanged in \([0,1)\). For \(M(z)=zg'(z)\), direct simplification gives \(M(z)=\frac{3}{32}(1-w^2)R(s)\), where \(R(s)=5(4a-5b+1)s^2+10(3b-1)s+5-4a-5b\).
The Bernstein coefficients of \(R\) are \(5-4a-5b\), \(10b-4a\), and \(16a\). They are nonnegative, and the first is positive, so \(R(s)>0\) for \(0\leq s<1\). The chain rule gives \(t^2\Phi''(t)=rM(z)[r\Psi(w)-1]\), where \(\Psi(w)=w\{(1-s)R'(s)-R(s)\}/R(s)\).

\emph{Separating the two sides of one.} Thus \(\Phi''(t)=0\) exactly when \(r\Psi(w)=1\). Write \(w=\varepsilon v\), where \(v=|w|\in[0,1)\) and \(\varepsilon\in\{-1,1\}\) records the side of \(z=1\). The curvature equation on that side is \(H_\varepsilon(v)=0\), where \(H_\varepsilon(v)=r\varepsilon v\{(1-v^2)R'(v^2)-R(v^2)\}-R(v^2)\) for \(0\leq v\leq1\).
Put \(d=4a+35b-15\) and \(e=16a+25b-15\), and recall that \(R(0)=5-4a-5b\). The degree-five Bernstein coefficients are
\[
\begin{alignedat}{3}
h_0^\varepsilon&=-R(0),\quad
&h_1^\varepsilon&=-R(0)+\varepsilon(r/5)d,\quad
&h_2^\varepsilon&=4a+2b-4+\varepsilon(2r/5)d,\\
h_3^\varepsilon&=4a-4b-2+\varepsilon(2r/5)e,\quad
&h_4^\varepsilon&=-8b+\varepsilon(16r/5)(6a-5b),\quad
&h_5^\varepsilon&=-16a(1+\varepsilon r).
\end{alignedat}
\]

\emph{Bounding the coefficient signs.} Maximizing over feasible \((r,b)\) gives
\begingroup\small\compactdisplays
\[
h_1^+>0\Rightarrow a>\tfrac5{12},\;
h_2^+>0\Rightarrow a>\tfrac5{14},\;
h_3^+>0\Rightarrow a>\tfrac{25}{52},\;
h_2^->0\Rightarrow a<\tfrac5{24},\;
h_3^->0\Rightarrow a<\tfrac{10}{41},\;
h_4^->0\Rightarrow a<\tfrac5{24}.
\]
\endgroup
For each displayed plus coefficient, positivity requires a positive coefficient on \(r\), so its positive maximum at fixed \(a\) occurs at \((r,b)=(1,1/2)\). For the displayed minus coefficients, positive maxima occur at \(r=1\), with \(b=a\) for \(h_2^-,h_3^-\) and \(b=1/2\) for \(h_4^-\). The remaining coefficients satisfy \(h_4^+\leq-24b/5\leq0\) and \(h_1^-<0\). At these maxima, \(h_1^+,h_2^+,h_3^+,h_2^-,h_3^-,h_4^-\) reduce, respectively, to \(24a/5-2\), \(28a/5-2\), \(52a/5-5\), \(2-48a/5\), \(4-82a/5\), and \(4-96a/5\).

\emph{At most two sign changes on each side.} The positive entries within each \mbox{sequence} form one block. For the \(+\) \mbox{sequence},
\(h_2^+-2h_1^+=6-4a-8b\geq0\), so \(h_1^+>0\) implies \(h_2^+>0\). If \(h_3^+>0\), the displayed bound gives \(a>25/52\), and direct subtraction gives
\(h_2^+-h_3^+\geq(26/5)a-2>0\); hence \(h_2^+>0\) as well. For the \(-\) sequence, put \(k=1/r\) and define \(h_2^-=(2r/5)L_2\), \(h_3^-=(2r/5)L_3\), and \(h_4^-=-(8r/5)L_4\).
Direct substitution gives \(L_2-L_3=L_4+5k(2b-1)\). If \(h_2^->0\) and \(h_4^->0\), then \(L_2>0>L_4\); since \(2b-1\leq0\), the identity gives \(L_3>L_2>0\), and hence \(h_3^->0\). Thus each coefficient sequence has at most two sign changes.

\emph{Curvature changes on only one side.} The \(+\) sequence can have a positive coefficient only if \(a>5/14\), the \(-\) sequence only if \(a<10/41\). Since \(10/41<5/14\), at most one coefficient sequence contains a positive entry. Without positive Bernstein coefficients, \(H_\varepsilon\) has no interior root; otherwise Bernstein sign variation gives at most two. As \(H_+\) and \(H_-\) are the curvature equations above and below \(t=1\), respectively, all curvature changes occur on one side, with at most two there.

\emph{A strict supporting tangent at one.} At \(r=1\), the tangent gap at one is affine in \((a,b)\). The admissible triangle \(0\leq a\leq b\leq1/2\) has vertices \((0,0),(0,1/2),(1/2,1/2)\), so checking them suffices. After removing \((t-1)^2\) and common positive factors, the three remaining polynomials are
\[
30t^4+148t^3+256t^2+44t+2,   \quad
 15t^4+58t^3+16t^2+134t+17,   \quad
 3t^4-14t^3+16t^2+14t+29.
\]
The first two have strictly positive coefficients, while the third equals
\(3(t^2-7t/3-1)^2+(17/3)t^2+26>0\). Thus the tangent to \(g\) at one is a strict global majorant. For \(0<r\leq1\) and \(t\ne1\), \(\Phi(t)=g(t^r)<g(1)+g'(1)(t^r-1)\leq\Phi(1)+\Phi'(1)(t-1)\), so the same tangent property holds for \(\Phi\).

\emph{Bounding the convexity interval.} Prize symmetry also gives \(\Phi''(1)=-\Phi'(1)<0\). The convexity interval cannot extend to infinity. Let \(n_{\min}\) be the first victory count with \(\widetilde q_{n_{\min}}>0\). For some \(C_{\mathrm{tail}}>0\), the small-ratio expansion and symmetry give \(1-\Phi(t)=\Phi(1/t)=C_{\mathrm{tail}}t^{-rn_{\min}}[1+o(1)]\) as \(t\to\infty\). Rationality in \(t^r\) permits term-by-term differentiation.
Hence \(\Phi''(t)=-C_{\mathrm{tail}}rn_{\min}(rn_{\min}+1)t^{-rn_{\min}-2}[1+o(1)]<0\) for all sufficiently large \(t\). With at most two curvature changes, all on one side, \(\Phi\) is either concave or has one connected convexity interval, wholly below or wholly above one.

\begin{samepage}
\emph{Locating the chord.} The same \(V_+-V_-\) argument produces at most one chord. If the convexity interval lies below one, the final concave piece contains \(t=1\). At slope \(\Phi'(1)\), the strict tangent at one gives \(V_+>V_-\); hence the joining slope is larger than \(\Phi'(1)\), forcing both contacts below one. If the convexity interval lies above one, the first concave piece contains \(t=1\), so the same tangent gives \(V_+<V_-\). The joining slope is then smaller than \(\Phi'(1)\), and the decreasing derivative on the first concave piece forces both contacts above one.\par
\end{samepage}

\emph{Strict reciprocal-deviation inequalities.} Reward symmetry places the reciprocal \(1/u\) of every positive chord endpoint \(u\) in an open strictly concave region where \(\operatorname{cav}\Phi=\Phi\), and gives \(\Phi'(1/u)=u^2\Phi'(u)\). Its tangent therefore supports the concave majorant globally. Equality at any other ratio would force the majorant to be affine on the segment joining the two points, contradicting strict concavity near \(1/u\). Thus the inequality is strict at every other ratio, including on the opposite side of one. The strict supporting tangent at \(1/u\) is exactly the reciprocal-deviation inequality. The same argument shows that the tangent at every positive ratio outside \([t_-,t_+]\), or at every positive ratio if there is no chord, touches \(\Phi\) only there. This proves the lemma.
\end{proof}

\begin{proof}[Proof of Theorem~\ref{thm:lowcount-refinements}]
For \(J=1,2\), admissibility yields \(\varphi(t)=t^r/(1+t^r)\), which is strictly concave for \(0<r\leq1\); the global tangent test and interchangeability show that equilibrium is unique and pure. For \(J\in\{3,4,5\}\), Lemma~\ref{lem:low-count-chord-geometry} gives at most one chord below one and strict reciprocal inequalities. If there is no chord, or if \(c\notin[t_-,t_+]\), Claim~\ref{prop:general-support}, strict best-response inequalities, and interchangeability show that equilibrium is unique and pure. At a chord endpoint, Proposition~\ref{prop:finite-contact} still applies: the mean condition forces a point mass despite the tangent's two contacts. Uniqueness follows from the mean, not tangent strictness. For \(c\in(t_-,t_+)\), the two-contact construction gives the unique semi-pure lottery, and Proposition~\ref{prop:finite-contact} makes it exhaustive. Because both contacts lie below one, only the higher-cost contestant can mix.

For \(J=6\), the lemma again gives at most one chord, now wholly below or wholly above one. A below-one chord yields the preceding \(A\)-pure branch; an above-one chord yields its role-reversed \(B\)-pure branch, in which the lower-cost contestant mixes. The two branches cannot coexist because there is only one chord. At either chord endpoint, the relevant mean ratio, \(c\) or \(1/c\), likewise forces a point mass. Away from the chord, the tangent inequalities are strict and imply that equilibrium is unique and pure.
\end{proof}

\subsection{Proofs for Section~\ref{S4:uniqueness}}
\label{app:finite-support-criteria}
\label{app:proof-uniqueness-criteria}
\label{app:no-uniform-support}

\begin{proof}[Proof of Corollary~\ref{cor:common-support-sets}]
Fix an equilibrium \((F_A^0,F_B^0)\), choosing one with both supports finite if available, and let \(D_i\) be contestant \(i\)'s positive best-response efforts against \(F_{-i}^0\). Interchangeability and continuity of expected payoff at positive efforts imply \(\operatorname{supp}F_i\cap(0,\infty)\subseteq D_i\) in every equilibrium. Bounded normalized prizes and the nonnegative payoff from inactivity give \(D_A\subset(0,1/c]\) and \(D_B\subset(0,1]\).

By Step~3 of Appendix~\ref{app:proof-general-taxonomy}, expected payoff against the fixed rival is analytic on \((0,\infty)\) and nonconstant. Its global maximizers have no positive accumulation point, so each bounded \(D_i\) is finite above every positive threshold and at most countable. If the fixed rival has finite support, Step~4 of the same appendix and Lemma~\ref{lem:generalized-polynomial-zeros} give finitely many positive stationary points. Every positive best response is such a point, so \(D_i\) is finite. Thus the positive supports of all equilibria lie in the common discrete sets \(D_A,D_B\), which are finite if one equilibrium has both supports finite. Finite- and infinite-support equilibria consequently cannot coexist. Since Proposition~\ref{thm:equilibrium-taxonomy} also makes pure and semi-pure supports finite, an infinite-support equilibrium forces all equilibria to be two-sided mixed with countably infinite positive supports.
\end{proof}

\begin{proof}[Proof of Proposition~\ref{prop:finite-contact}]
First, \(\Gamma(\sigma)\) is finite for every \(\sigma>0\). Every positive contact solves \(\varphi'(t)=\sigma\). Multiplying this equation by \((1+t^r)^{J+1}\) gives a finite sum of real powers that is not identically zero: otherwise \(\varphi'(t)\equiv\sigma>0\), contradicting bounded reward. Lemma~\ref{lem:generalized-polynomial-zeros} gives finitely many positive solutions; zero adds at most one contact.

Fix \(F_{T_B}\) supported on \(\Gamma(\sigma)\) with mean \(c\). Let \(F_B\) denote the induced distribution of \(X_B=\sigma T_B\). Against \(X_A=\sigma\), contestant~\(B\)'s payoff depends on her effort ratio \(t=X_B/\sigma\) through \(\varphi(t)-\sigma t\), so every distribution supported on \(\Gamma(\sigma)\) is a best response. It remains to check contestant~\(A\).

At effort \(z\sigma\), \(z>0\), \(A\)'s normalized payoff is \(\Pi_A(z\sigma;F_B)=F_{T_B}(\{0\})+\int_{u>0}\varphi(z/u)\,dF_{T_B}(u)-c\sigma z\); against a positive contact \(u\), her own-to-rival ratio is \(z/u\).
The strict pure-contestant best-response test gives \(\varphi(z/u)-u\sigma z<\varphi(1/u)-u\sigma\) whenever \(z\neq1\). Averaging and using \(\int u\,dF_{T_B}(u)=c\) gives \(\Pi_A(z\sigma;F_B)<\Pi_A(\sigma;F_B)\). Thus \(\sigma\) is \(A\)'s unique positive best response. Inactivity is also strictly worse: evaluating the same inequality at zero gives \(\varphi(1/u)-u\sigma>0\) at every positive contact. Positive mean \(c>0\) ensures positive mass on such contacts, so \(\Pi_A(\sigma;F_B)=F_{T_B}(\{0\})+\int_{u>0}[\varphi(1/u)-u\sigma]\,dF_{T_B}(u)>F_{T_B}(\{0\})/2\), the payoff from inactivity. Lemma~\ref{lem:core-reductions}\emph{(i)} therefore rules out nonuniform positive-total deviations by either contestant. Hence \((X_A,X_B)=(\sigma,\sigma T_B)\), \(T_B\sim F_{T_B}\), is an equilibrium of the multi-battle contest.

Choose one such ratio distribution \(F_{T_B}^0\) and its induced effort distribution \(F_B^0\). For any other equilibrium \((F_A,F_B)\), interchangeability pairs \(F_A\) with \(F_B^0\). The unique-best-response conclusion forces \(F_A\) to concentrate on \(\sigma\). The support of contestant~\(B\)'s ratio \(T_B=X_B/\sigma\) must then lie in \(\Gamma(\sigma)\). Finally, the cost-balance identity \(c\mathbb E[X_A]=\mathbb E[X_B]\), with \(X_A=\sigma\), gives \(\mathbb E[T_B]=c\), as in Claim~\ref{prop:general-semi-pure}. This completes the characterization. The \(B\)-pure case is identical after exchanging roles: against a positive contact \(\rho\) of \(A\), the test bounds \(B\)'s reward change from effort \(s\sigma/c\) by \(\rho\sigma(s-1)\), which averages to \(\sigma(s-1)\mathbb E[T_A]=(s-1)\sigma/c\), exactly her cost change; Proposition~\ref{thm:equilibrium-taxonomy}\emph{(i)} excludes zero from \(A\)'s support.
 \end{proof}

\subsection{Proofs for Section~\ref{S4:examples}}\label{app:count-examples}\label{app:cross-count-identities}

\begin{proof}[Proof of Lemma~\ref{lem:cross-count-identities}]
 Start with \(J+1\) independent battle outcomes with common win probability \(p\). Independently discard a uniformly chosen battlefield and apply \(\boldsymbol q\) to the remaining \(J\). Given \(n\) wins (\(1\leq n\leq J\)), deletion leaves \(n-1\) wins with probability \(n/(J+1)\) and \(n\) otherwise, giving conditional expected prize \(q_n^+\). The retained trials remain independent with win probability \(p\), preserving expected prizes. The elevation preserves endpoints, while \(q_n^+\leq q_n\leq q_{n+1}^+\) for \(0\leq n\leq J\) and \(q_{J+1-n}^+=1-q_n^+\) give monotonicity and symmetry; thus \(\boldsymbol q^+\) is admissible.

 \end{proof}

\section{Proofs for Section~\ref{S5}}\label{app:majority-proofs}

\subsection{Proofs for Section~\ref{S5:infinite-phase}}\label{app:proof-high-elasticity}

 \begin{proof}[Proof of Theorem~\ref{thm:infinite-phase}]
 \emph{Recovering equal costs from any infinite-support equilibrium.}
 Take an equilibrium with infinite positive support at some \(c\). Proposition~\ref{thm:equilibrium-taxonomy} makes both positive supports countably infinite, accumulating only at zero, and gives \(F_A(\{0\})=0\). Write \(\zeta=F_B(\{0\})\), \(p=1-\zeta>0\), and \(\chi=c\mathbb E[X_A]=\mathbb E[X_B]\). Indifference along positive support points tending to zero and bounded convergence give normalized payoffs \(\widehat\pi_A=\zeta\) and \(\widehat\pi_B=0\). Prize balance then gives \(\zeta=1-2\chi\), so \(p=2\chi\).

Let \(H\) be \(B\)'s conditional positive-effort distribution. If \(A\) independently copies \(H\), her expected prize is \(\zeta+p/2\), and her expected cost is \(c\chi/p=c/2\). Optimality therefore gives \(\zeta\geq\zeta+(p-c)/2\), or \(p\leq c\). If \(B\) independently copies \(F_A\), her payoff is \(1/2-\chi/c=1/2-p/(2c)\). Her zero equilibrium payoff gives \(p\geq c\). Thus
\[
p=c,\qquad \chi=c/2,\qquad \mathbb E_{F_A}[X]=\mathbb E_H[X]=1/2.
\]
For any distribution \(F\) with no atom at zero, define \(V_F(x)=\int\mathcal U(x,y)\,dF(y)-x\). Contestant \(A\)'s positive deviation inequalities become \(1-c+cV_H(x)\leq1-c\), with equality \(F_A\)-almost surely. 

 Contestant \(B\)'s deviation inequalities are \(V_{F_A}(y)\leq0\), with equality \(H\)-almost surely. The corresponding equal-cost deviation inequalities also hold at zero, with equality, because neither distribution has an atom there. Hence \((F_A,H)\) is an equal-cost equilibrium with infinite positive supports.

\emph{The forward correspondence.}
Suppose an equal-cost equilibrium has infinite support. Corollary~\ref{cor:common-support-sets} makes every equal-cost equilibrium have infinite support. Every \(F\in\mathcal S(\varphi)\) has no atom at zero and support points tending to zero. Bounded convergence and indifference give zero normalized rent, and symmetry gives \(\mathbb E_F[X]=1/2\). Thus \(V_F\leq0\) everywhere. Interchangeability gives equality against every distribution in \(\mathcal S(\varphi)\).

 For \(F_1,F_2\in\mathcal S(\varphi)\), let \(F_A=F_1\), \(F_B=(1-c)\delta_0+cF_2\). Every positive deviation by \(A\) gives \(1-c+cV_{F_2}(x)\leq1-c\), with equality \(F_1\)-almost surely, while inactivity gives \((1-c)/2\). Contestant \(B\)'s payoff is \(V_{F_1}(y)\leq0\), with equality at zero and \(F_2\)-almost surely. These pairs are infinite-support equilibria at every \(c\), and Corollary~\ref{cor:common-support-sets} excludes any finite-support equilibrium at those parameters. This proves the equivalence of the three phase conditions.

\emph{Exhaustiveness.}
 The constructed pairs and Proposition~\ref{prop:outcome-equivalence} fix every equilibrium's mean efforts at \((1/2,c/2)\) and payoffs at \((1-c,0)\). In any equilibrium, \(B\)'s zero payoff gives \(V_{F_A}\leq0\). Two independent draws from \(F_A\), together with prize symmetry and the mean \(1/2\), give
\[
\int V_{F_A}(x)\,dF_A(x)=\tfrac12-\mathbb E[X_A]=0.
\]
 Hence \(F_A\in\mathcal S(\varphi)\). 

 Its support approaches zero, so the limit of \(A\)'s indifference condition gives \(F_B(\{0\})=1-c\). Write \(F_B=(1-c)\delta_0+cF_2\). The residual distribution has no atom at zero and mean \(1/2\); \(A\)'s deviation inequalities give \(V_{F_2}\leq0\). Integrating against \(F_2\) again gives zero, implying \(F_2\in\mathcal S(\varphi)\). This proves~\eqref{eq:infinite-phase-set} and the remaining support and outcome statements.
\end{proof}

\begin{proof}[Proof of Proposition~\ref{prop:high-elasticity}]
Take \(F\in\mathcal S(\varphi)\). If its support were bounded away from zero, it would have a smallest point \(x_*>0\). Indifference and the first-order condition there give
\[
\int\varphi(x_*/y)\,dF(y)=x_*+\widehat\pi^*,\qquad
\int(x_*/y)\varphi'(x_*/y)\,dF(y)=x_*,
\]
where \(\widehat\pi^*\geq0\) is the equilibrium payoff. Since \(0<x_*/y\leq1\), condition~\eqref{eq:high-elasticity} makes the second integrand strictly larger than the first, a contradiction. Thus the support approaches zero, and Proposition~\ref{thm:equilibrium-taxonomy} makes it countably infinite. Theorem~\ref{thm:infinite-phase} applies.
 \end{proof}

\subsection{Proofs for Section~\ref{S5:majority}}\label{app:majority-uniqueness}

\begin{lemma}[Equal-cost uniqueness under majority rule]\label{lem:majority-kernel}
For every \(J\geq1\) and \(r\in(0,1]\), the scalar majority contest has a unique equal-cost equilibrium.
\end{lemma}
\begin{proof}
We compare two equilibrium distributions through their marginal rewards; the integral representation below forces their difference to vanish.

Put \(k=\lceil J/2\rceil\), \(b=(2k-1)\binom{2k-2}{k-1}\), and \(\varphi=\varphi_{J,r}^{\mathrm{maj}}\). The odd--even reward identity in Section~\ref{S5:majority} gives the effort-weighted marginal reward
\[
m(x,y)=r b\frac{x^{rk}y^{rk}}{(x^r+y^r)^{2k}}
\qquad(x,y>0);
\]
set \(m=0\) when either effort is zero. It is symmetric and bounded by \(r b/4^k\).

Take \(F,G\in\mathcal S(\varphi)\), a nonempty set by existence, symmetry, and interchangeability. Neither distribution has an atom at zero, and domination confines their supports to \([0,1]\). Interchangeability makes every positive point in their combined support a best response to both distributions. Boundedness of \(m\) justifies differentiation under the rival distribution locally at positive own efforts. The first-order conditions therefore give
\[
\int m(x,y)\,dF(y)=x=\int m(x,y)\,dG(y)
\quad\bigl(x\in(\operatorname{supp}F\cup\operatorname{supp}G)\cap(0,\infty)\bigr).
\]
Consequently, the signed measure \(\mu=F-G\), the difference between the distributions, satisfies \(\iint m\,d\mu\,d\mu=0\).

For \(a>0\) and integers \(n\geq1\), the identity \(\int_0^\infty s^{n-1}e^{-as}\,ds=(n-1)!/a^n\) gives
\[
m(x,y)=\frac{r b}{(2k-1)!}\int_0^\infty
s^{2k-1}[x^{rk}e^{-sx^r}][y^{rk}e^{-sy^r}]\,ds.
\]
Boundedness of \(m\) and finiteness of \(|\mu|\) give absolute integrability, so Fubini's theorem permits exchanging the integrals and yields
\[
0=\frac{r b}{(2k-1)!}\int_0^\infty s^{2k-1}
\left[\int x^{rk}e^{-sx^r}\,d\mu(x)\right]^2\,ds.
\]
Continuity makes the inner integral zero for every \(s>0\). Change variables to \(u=x^r\) and weight the resulting signed measure by \(u^k\). This finite signed measure on \([0,1]\) integrates \(e^{-su}\) to zero for every \(s>0\). Differentiating at \(s=0\) gives zero integrals against every power \(u^n\). Since polynomials uniformly approximate continuous functions, the measure is zero. Dividing out \(u^k\) on intervals away from zero and using \(\mu(\{0\})=0\) gives \(\mu=0\). Thus \(F=G\); symmetry and interchangeability imply uniqueness of the entire equal-cost equilibrium.
\end{proof}

\begin{proof}[Proof of Theorem~\ref{thm:majority-regimes}\emph{(iii)}]
Lemma~\ref{lem:majority-kernel} makes \(\mathcal S_{J,r}\) a singleton. In the stated region, Proposition~\ref{prop:high-elasticity} and Theorem~\ref{thm:infinite-phase} give the unique unequal-cost pair and its support and mean. Proposition~\ref{thm:existence} transfers the result to the unrestricted game.
\end{proof}

\section*{Declaration of generative AI and AI-assisted technologies in the manuscript preparation process}
During the preparation of this work, the authors used ChatGPT and Claude to improve the language and to assist in checking the mathematical arguments. After using these tools, the authors reviewed and edited the content as needed and take full responsibility for the content of the published article.

\begingroup
\setlength{\bibsep}{0pt}
\ifx\undefined\bysame

  \newcommand{\bysame}{\leavevmode\hbox to3em{\hrulefill}}
 \fi

\endgroup

\clearpage
\setcounter{page}{1}
\renewcommand{\thepage}{OA-\arabic{page}}
\begingroup
\makeatletter
\def\@title{Online Appendix for ``Equilibrium Architecture in \\ Multi-Battle Contests with Count-Dependent Prizes''}
\let\@thanks\@empty
\maketitle
\makeatother
\endgroup
\makeatletter
\@removefromreset{theorem}{section}
\@removefromreset{proposition}{section}
\@removefromreset{claim}{section}
\@removefromreset{corollary}{section}
\@removefromreset{lemma}{section}
\@removefromreset{assumption}{section}
\@removefromreset{definition}{section}
\@removefromreset{example}{section}
\@removefromreset{remark}{section}
\@removefromreset{equation}{section}
\@removefromreset{figure}{section}
\@removefromreset{table}{section}
\makeatother
\setcounter{section}{0}
\setcounter{theorem}{0}
\setcounter{proposition}{0}
\setcounter{claim}{0}
\setcounter{corollary}{0}
\setcounter{lemma}{0}
\setcounter{assumption}{0}
\setcounter{definition}{0}
\setcounter{example}{0}
\setcounter{remark}{0}
\setcounter{equation}{0}
\setcounter{figure}{0}
\setcounter{table}{0}
\renewcommand{\thesection}{OA.\arabic{section}}
\renewcommand{\thetheorem}{OA.\arabic{theorem}}
\renewcommand{\theproposition}{OA.\arabic{proposition}}
\renewcommand{\theclaim}{OA.\arabic{claim}}
\renewcommand{\thecorollary}{OA.\arabic{corollary}}
\renewcommand{\thelemma}{OA.\arabic{lemma}}
\renewcommand{\theassumption}{OA.\arabic{assumption}}
\renewcommand{\thedefinition}{OA.\arabic{definition}}
\renewcommand{\theexample}{OA.\arabic{example}}
\renewcommand{\theremark}{OA.\arabic{remark}}
\renewcommand{\theequation}{OA.\arabic{equation}}
\renewcommand{\thefigure}{OA.\arabic{figure}}
\renewcommand{\thetable}{OA.\arabic{table}}

\noindent This online appendix contains supplementary calculations and proofs. References to sections, results, equations, and figures without the prefix ``OA'' are to the main paper. As in the main paper, FKL refers to \citet*{fan2026generalized}.

\section{Supplement to Section~\ref{S4:examples}}\label{OA:small-counts}

This section gives equilibrium cutoffs for general three- and four-battlefield schedules and a five-battlefield example in which the higher-cost contestant mixes between two positive efforts at \(r=1\). It also verifies Examples~\ref{ex:six-lower-cost-mixing}--\ref{ex:finite-two-sided}. It uses the notation \(\mathcal B_K\) and degree elevation from Section~\ref{S4:examples}.

\subsection{Degree Reduction}\label{OA:degree-reduction}
For even \(J\geq2\), define the \emph{degree reduction} \(\boldsymbol q^-\) by
 \[
q_0^-=q_0,\qquad
q_n^-=\frac{Jq_n-nq_{n-1}^-}{J-n}
\quad(n=1,\ldots,J-1).
\]
 
  \begin{lemma}[Degree reduction] \label{lem:degree-reduction}
Let \(J\) be even. If \(\boldsymbol q\) is symmetric, then \(\mathcal B_{J-1}[\boldsymbol q^-]=\mathcal B_J[\boldsymbol q]\); however, \(\boldsymbol q^-\) need not be admissible.
\end{lemma}

Reduction can violate admissibility: the four-battlefield schedule \(\boldsymbol q=(0,\tfrac12,\tfrac12,\tfrac12,1)\) reduces to \(\boldsymbol q^-=(0,\tfrac23,\tfrac13,1)\), which is not monotone.

\begin{proof}
Reverse the averaging in the proof of Lemma~\ref{lem:cross-count-identities}: \(Jq_n=nq_{n-1}^-+(J-n)q_n^-\) gives the recursion for elevation from \(J-1\) to \(J\). Starting from \(q_0^-=q_0\), solve it for \(1\leq n<J/2\), then reflect: \(q_{J-1-n}^-=1-q_n^-\). The forward recursion defining \(\boldsymbol q^-\) uniquely solves the equations for \(n=0,\ldots,J-1\), so it coincides with this reflected solution. The middle equation holds because \(q_{J/2}=1/2\); symmetry \(q_{J-n}=1-q_n\) makes the upper-half equations mirror the lower half, with endpoint \(q_{J-1}^-=1-q_0=q_J\). Thus all elevation equations hold, proving the identity.
\end{proof}

\subsection{Three- and Four-Battlefield Cutoffs}\label{OA:three-four}
Write \(\widetilde{\boldsymbol q}_{3,\omega}=(0,1-\omega,\omega,1)\), \(\omega\in[1/3,1]\), with scalar reward \(\varphi_{3,\omega,r}\). This is an admissible three-battlefield schedule exactly when \(\omega\geq1/2\). By Lemma~\ref{lem:degree-reduction}, the four-battlefield schedule \(\widetilde{\boldsymbol q}_{4,a}=(0,a,1/2,1-a,1)\), \(a\in[0,1/2]\), has the same scalar reward as \(\widetilde{\boldsymbol q}_{3,\omega}\) at \(\omega=1-4a/3\). The reduction is admissible exactly when \(a\leq3/8\); otherwise \(2\omega-1<0\) and the concavity criterion below applies.

For both counts, \(r(2\omega-1)\leq1/2\) gives strict reward concavity and a unique pure equilibrium for every \(c\), in particular whenever \(r\leq1/2\) or \(\omega\leq3/4\). When either \(r=1\) and \(\omega>3/4\), or \(\omega=1\) and \(r>1/2\), the supporting chord joins zero to \(c_+(\omega,r)>0\): \(B\) mixes between inactivity and one positive effort for \(c<c_+\), and equilibrium is pure otherwise.

For \(r<1\) and \(\omega<1\), any supporting chord has two positive contacts \(0<c_-(\omega,r)<c_+(\omega,r)<1\). The unique equilibrium is semi-pure for \(c_-<c<c_+\), with \(B\) mixing between the corresponding efforts, and pure otherwise, including at the cutoffs or without a chord.

For example, at \(r=4/5\) and \(\omega=0.9\), \((0,0.1,0.9,1)\) has cutoffs \(c_-\approx0.012985\), \(c_+\approx0.083530\). Its elevation \((0,0.075,1/2,0.925,1)\) has the same reward and pure--semi-pure--pure classification as \(c\) increases, beyond the \(r=1\) three-battlefield benchmark of \citet{luwangzhou2024electoral}. The derivation below gives cutoff formulas and exact contact equations; decimals are rounded.

\begin{proof}[Derivation]

Under \(r(2\omega-1)\leq1/2\), the second Bernstein coefficient in the proof of Lemma~\ref{lem:low-count-chord-geometry} is \([4r(2\omega-1)-2]/3\leq0\); the others are nonpositive and the last is negative. Thus \(\varphi''<0\) on \((0,1)\); the lemma's strict concavity on \([1,\infty)\) gives global strict concavity. The strict tangent test (Claim~\ref{prop:general-support} with strict inequalities; Section~\ref{S4:uniqueness}) yields a unique pure equilibrium for every \(c\). This includes four battlefields, with \(\alpha=1-\omega=4a/3\in[0,2/3]\) in the notation of that proof, even when the reduced schedule is not admissible.

Now set \(r=1\) and suppose \(\omega>3/4\). At a positive tangency point \(t\), the gain over inactivity has the sign of \(R_\omega(t)=t^2+(6\omega-2)t+9-12\omega\), since \(\varphi(t)-t\varphi'(t)=t^2R_\omega(t)/(1+t)^4\).
Its unique positive root is \(c_+(\omega,1)=-(3\omega-1)+\sqrt{9\omega^2+6\omega-8}\).
Since \(d[\varphi(t)/t]/dt=-[\varphi(t)-t\varphi'(t)]/t^2\), the ratio \(\varphi(t)/t\) rises up to \(c_+(\omega,1)\) and falls thereafter. Hence the unique supporting chord satisfying \(\operatorname{cav}\varphi>\varphi\) in its interior joins the origin to the positive contact. The two-contact case of Proposition~\ref{prop:finite-contact} gives mixing between inactivity and that positive contact exactly when \(c<c_+(\omega,1)\); otherwise the equilibrium is pure.

Now let \(r\in(0,1)\) and write \(\varphi=\varphi_{3,\omega,r}\). Direct differentiation gives \(\varphi'(t)=3r t^{r-1}\big[(1-\omega)+2(2\omega-1)t^r+(1-\omega)t^{2r}\big]/(1+t^r)^4\). If \(\omega<1\), then \(\varphi(t)/t\sim3(1-\omega)t^{r-1}\to\infty\) as \(t\downarrow0\), ruling out a finite-slope supporting chord from zero. By Lemma~\ref{lem:low-count-chord-geometry}, any nontrivial chord therefore has contacts \(0<c_-<c_+<1\). They satisfy \(\varphi'(c_-)=\varphi'(c_+)\) and \(\varphi(c_-)-c_-\varphi'(c_-)=\varphi(c_+)-c_+\varphi'(c_+)\), equating the tangents' slopes and intercepts. When a chord exists, its contacts are the unique distinct solution of these equations with global support: \(\varphi(t)\leq\varphi(c_-)+\varphi'(c_-)(t-c_-)\) for every \(t\geq0\). The mean-ratio condition gives upper-contact probability \((c-c_-)/(c_+-c_-)\); together with the lemma's strict reciprocal inequalities, this yields the unique semi-pure equilibrium inside the cutoff interval and pure equilibrium at its endpoints. Outside the interval, or if no chord exists, the same lemma shows that equilibrium is unique and pure. For the illustration \((r,\omega)=(4/5,9/10)\), \(\varphi''(1/32)=12091392/35496425>0\), so a two-positive-contact chord exists.

For majority rule, \(\omega=1\), reward elasticity is \(t\varphi'(t)/\varphi(t)=6r/[(1+t^r)(3+t^r)]\). For \(r>1/2\), this decreases through one at \(t^r=\sqrt{1+6r}-2\). Hence \(\varphi(t)/t\) has its unique maximum there, proving that the origin chord is globally supporting and giving \(c_-(1,r)=0\) and \(c_+(1,r)=(\sqrt{1+6r}-2)^{1/r}\).

Lemma~\ref{lem:degree-reduction} transfers the same conclusion to four battlefields.
\end{proof}

\subsection{Five-Battlefield Mixing}\label{OA:five}\label{app:five-boundary-algebra}
For \(r\leq1/3\), equilibrium is unique and pure. At \(r=1\), semi-pure equilibria with three or four battlefields require inactivity; five battlefields can instead support mixing between two positive efforts. For the schedule \((0,0.11,0.11,0.89,0.89,1)\), the unique chord has contacts \(t_-\approx0.0082\) and \(t_+\approx0.6409\). At \(c=0.20\), the higher-cost contestant mixes between the corresponding efforts. The derivation below gives exact polynomial systems for the \(r=1\) concavity boundary, the transition from mixing with inactivity to mixing between two positive efforts, and the contacts.

\begin{proof}[Derivation]
Set \(r=1\), write \(a=\widetilde q_1\), \(b=\widetilde q_2\), and let \(\varphi=\varphi_{5,1}\), where \(0\leq a\leq b\leq1/2\). Define
\[
\widehat h_{a,b}(t)\equiv\frac{\varphi'(t)}{5}
=\frac{at^4+4(b-a)t^3+6(1-2b)t^2+4(b-a)t+a}{(1+t)^6}.
\]

Since \(\varphi''(t)=-10\mathcal C_{a,b}(t)/(1+t)^7\), where \(\mathcal C_{a,b}(t)=at^4+(6b-8a)t^3+(12-30b+6a)t^2+(22b-10a-6)t+5a-2b\),
Lemma~\ref{lem:low-count-chord-geometry} gives three cases. If \(\mathcal C_{a,b}(t)\geq0\) for every \(t\geq0\), then \(\operatorname{cav}\varphi=\varphi\). An origin chord's positive contact \(t_+\in(0,1)\) maximizes \(\varphi(t)/t\) and satisfies \(R_{a,b}(t_+)=0\), where \(R_{a,b}(t)=t^4+(6-10a)t^3+(15+15a-30b)t^2+(60b-20)t+25a-10b\).
For \(t>0\), \(R_{a,b}(t)=0\) is equivalent to the origin-chord tangency condition \(t\varphi'(t)=\varphi(t)\).
The root in \((0,1)\) is uniquely selected by requiring the line through \((0,0)\) and \((t_+,\varphi(t_+))\) to lie weakly above \(\varphi\) everywhere. In the remaining nonconcave case, two positive contacts \(0<t_-<t_+<1\) form the unique distinct solution of \(\widehat h_{a,b}(t_-)=\widehat h_{a,b}(t_+)=[\varphi(t_+)-\varphi(t_-)]/[5(t_+-t_-)]\) whose line lies weakly above \(\varphi\) everywhere.

At the concavity boundary, an interior minimum of \(\mathcal C_{a,b}\) just reaches zero, giving \(\mathcal C_{a,b}(t)=0,\ \mathcal C_{a,b}'(t)=0,\ t\in(0,1)\). A minimum at \(t=0\) gives the boundary segment \(b=\frac{5}{2}a,\ a\in[\frac{2}{15},\frac{1}{5}]\). The transition from an origin to a two-positive-contact chord is \(R_{a,b}(t)=0,\ \widehat h_{a,b}(t)=a,\ t\in(0,1)\); the second equation marks the lower contact's departure from zero.

For fixed \(t\), each system is linear in \((a,b)\), giving a rational one-parameter description wherever its determinant is nonzero.
For \(a=b=11/100\), \(\varphi''(1/5)=60725/93312>0\). Also,
\[
\frac{11}{20}t-\varphi(t)=\frac{t^2Q_{\mathrm{gap}}(t)}{20(1+t)^5},\qquad
Q_{\mathrm{gap}}(t)=33-68t+21t^2+35t^3+11t^4.
\]
\begin{samepage}
The Bernstein coefficients of \(Q_{\mathrm{gap}}\) on \([0,1]\) are \((33,16,5/2,5/4,32)\), and \(Q_{\mathrm{gap}}(1+v)=32+123v+192v^2+79v^3+11v^4\). Thus \(Q_{\mathrm{gap}}(t)>0\) for \(t\geq0\), excluding an origin chord.\par
\end{samepage}
Lemma~\ref{lem:low-count-chord-geometry} therefore gives \(0<t_-<1/5<t_+<1\) and verifies all deviations. The reported decimals approximate these exact contacts.
\end{proof}

\subsection{Verification of Example~\ref{ex:six-lower-cost-mixing}}\label{app:six-lower-cost-mixing}
\begin{proof}[\unskip\nopunct]
Set \(r=1\) and \(\widetilde{\boldsymbol q}=(0,\tfrac12,\tfrac12,\tfrac12,\tfrac12,\tfrac12,1)\). Then
\[
\varphi_6(t)=\frac{1}{2}+\frac{t^6-1}{2(1+t)^6},\qquad
\varphi_6''(t)=-\frac{3Q_{\mathrm{curv}}(t)}{(1+t)^7},\quad
Q_{\mathrm{curv}}(t)=2t^4-7t^3+7t^2-7t+7.
\]
Because \(Q_{\mathrm{curv}}(1)=2\), \(Q_{\mathrm{curv}}(2)=-3\), and \(Q_{\mathrm{curv}}(3)=22\), the unique convexity interval of Lemma~\ref{lem:low-count-chord-geometry} lies above one and contains two. Its supporting chord has contacts \(1<t_-<2<t_+\). The strictly concave left branch starts at the origin, so its tangent at \(t_->0\) has positive intercept \(\varphi_6(t_-)-t_-\varphi_6'(t_-)\). Numerically, \(t_-\approx1.06518\), \(t_+\approx3.06644\), with slope \(\sigma\approx0.0444009\). At \(c=1/2\), take \(X_B=\sigma/c=2\sigma\) and let \(T_A\) assign probability \(\lambda=(2-t_-)/(t_+-t_-)\) to \(t_+\) and probability \(1-\lambda\) to \(t_-\). Then \(\mathbb E[T_A]=2=1/c\). The role-reversed two-contact construction and the lemma's strict reciprocal inequalities verify the equilibrium in Example~\ref{ex:six-lower-cost-mixing}; the positive intercept makes inactivity strictly worse for \(A\), and Proposition~\ref{prop:finite-contact} gives uniqueness.
\end{proof}

\subsection{Verification of Example~\ref{ex:seven-multiplicity}}\label{app:seven-multiplicity-proof}
\begingroup
\begin{proof}[\unskip\nopunct]
Let \(\sigma=612/1075\), \(t_1=1/9\), and \(t_2=1/4\). Substituting the schedule from the example gives
\[
\sigma t-\varphi(t)
=\frac{3t(t-t_1)^2(t-t_2)^2(1428t^3+8519t^2+19729t+20388)}{7525(1+t)^7}\geq0.
\]
Hence the supporting line touches the reward curve exactly at \(0,t_1,t_2\). For each positive contact \(u\in\{t_1,t_2\}\), the reciprocal-deviation gap factors as
\[
[\varphi(1/u)-u\sigma]-[\varphi(z)-u^2\sigma z]
=\frac{(z-1/u)^2\mathcal P_u(z)}{(1+z)^7},
\]
where direct expansion gives degree-six polynomials \(\mathcal P_u\) whose coefficients are strictly positive for \(u\in\{1/9,1/4\}\).\footnote{In ascending powers, \(\mathcal P_{1/9}\): \((313/29025\), \(6413/90300\), \(41519/216720\), \(188611/650160\), \(7657/54180\), \(476/9675\), \(68/9675)\); \(\mathcal P_{1/4}\): \((769/17200\), \(327631/1083600\), \(548711/650160\), \(879127/650160\), \(10036/13545\), \(1071/4300\), \(153/4300)\).}
The inactivity comparisons are also strict:
\(\varphi(1/t_1)-t_1\sigma=939/1075>0\) and
\(\varphi(1/t_2)-t_2\sigma=769/1075>0\).
For \(0<c\leq1/4\), Proposition~\ref{prop:finite-contact} applies. Every distribution on \(\{0,t_1,t_2\}\) with mean \(c\) generates an equilibrium with \(X_A=\sigma\) and \(X_B=\sigma T_B\), and these are all equilibria. For \(c\in(0,1/4)\), the total-mass and mean conditions leave one probability free, producing the stated line segment.
\end{proof}
\endgroup

\subsection{Verification of Example~\ref{ex:finite-two-sided}}\label{app:finite-two-sided}
\begin{proof}[\unskip\nopunct]
Here \(\varphi(t)=\tfrac12+(t^7-1)/[2(1+t)^7]\) and \(\varphi'(t)=7(t^6+1)/[2(1+t)^8]\). Let \(f(s)=121s^6-326s^5-457s^4-468s^3-457s^2-326s+121\).
Exact evaluation gives \(f(3.963772)<0<f(3.963773)\); continuity yields a root \(s\in(3.963772,3.963773)\). Define
\[
a=\frac{\sum_{k=0}^6s^k}{2(1+s)^7},\qquad
p=\frac{128s^2(1+s^6)-(1+s)^8}{(1+s)[128s(1+s^6)-(1+s)^8]}.
\]
Rational bounds give a positive denominator for \(p\) and \(0.9360<p<0.9361\), \(0.03492<a<0.03494\), and \(0.1384<sa<0.1385\). Thus \(F=p\delta_a+(1-p)\delta_{sa}\) is a two-point distribution.

Against \(F\), write a deviation as \(x=az\) and its normalized payoff as \(H(z)=p\varphi(z)+(1-p)\varphi(z/s)-az\). Substitution and symmetry give \(H(1)=H(s)\) and \(H'(1)=H'(s)=-(s-1)^2f(s)/[256(1+s)^9]=0\). Solving \(f(s)=0\) for \(s^6\) and substituting for higher powers gives
\[
H(1)-H(z)=
\frac{(z-1)^2(z-s)^2P_s(z)}{256(1+s)^9(1+z)^7(s+z)^7P_{\mathrm{den}}(s)},
\qquad z\geq0,
\]
where \(P_{\mathrm{den}}(s)=118s^5+209s^4+244s^3+209s^2+118s-s^6-1\) and
\(P_s(z)=\sum_{j=0}^{11}P_j(s)z^j\). Because \(3<s<4\),
\(P_{\mathrm{den}}(s)>s^5(118-s)-1>0\). It remains to prove \(H(1)-H(z)\geq0\) by showing that every coefficient \(P_j(s)\) is positive. Put
\(u=s+s^{-1}\). Reducing again with \(f(s)=0\) gives
\[
P_j(s)=\kappa_j(1+s)^{e_j}s^{m_j}
\{\alpha_j+\beta_j(u-4)+\gamma_j(u^2-16)\},
\qquad j=0,\ldots,11,
\]
with the exact coefficients below.
\par\smallskip
\begingroup
\footnotesize
\setstretch{1}
\setlength{\tabcolsep}{5pt}
\renewcommand{\arraystretch}{1.08}
\noindent\fitwidth{\begin{tabular}{@{}c c c c r r r@{}}
\hline
\(j\) & \(\kappa_j\) & \(e_j\) & \(m_j\) & \(\alpha_j\) & \(\beta_j\) & \(\gamma_j\) \\
\hline
0  & \(2^{15}/11^5\)              & 1 & 12 & 2296278280             & 161131318            & 105562027 \\
1  & \(2^{14}/11^8\)              & 2 & 11 & 19903364471479         & 1399614783092        & 913661626344 \\
2  & \(2^{14}/11^{10}\)           & 1 & 11 & 24892755778329290      & 1747536442747027     & 1143992338029810 \\
3  & \(2^{14}/11^{10}\)           & 2 & 10 & 27326002340004839      & 1919558701412938     & 1255286937939934 \\
4  & \(2^{15}/11^{12}\)           & 1 & 10 & 7939003325423267016    & 556693721888089286   & 365135686149414627 \\
5  & \(7\cdot2^{14}/11^{12}\)    & 2 & 9  & 1219165548466875322    & 85476882483071726    & 56078212833017721 \\
6  & \(2^{15}/11^{14}\)           & 1 & 9  & 1181345679926815479726 & 82755674127902547017 & 54369282748763887378 \\
7  & \(2^{14}/11^{12}\)           & 2 & 8  & 4864564780862868374    & 340525088863699482   & 223991808676661219 \\
8  & \(2^{16}/11^{12}\)           & 1 & 8  & 1249602821211452616    & 87449653981671262    & 57549276770474991 \\
9  & \(2^{14}/11^{10}\)           & 2 & 7  & 4343163929791991       & 303511854826142      & 200210318528128 \\
10 & \(3\cdot2^{14}/11^{10}\)    & 1 & 7  & 547834131206542        & 38314702022305       & 25240471208550 \\
11 & \(2^{14}/11^8\)              & 2 & 6  & 371837542151           & 26004936400          & 17132122038 \\
\hline
\end{tabular}}
\par\smallskip
\endgroup
\noindent The bound \(3.963772<s<3.963773\) implies \(u>4\). All table entries \(\kappa_j,\alpha_j,\beta_j,\gamma_j\) are positive, so \(P_j(s)>0\) for every \(j\) and \(P_s(z)>0\) for every \(z\geq0\). Hence \(H(z)\leq H(1)=H(s)\) for all \(z\geq0\), with equality only at \(1,s\); the only best responses to \(F\) are \(a,sa\). Both contestants face this problem because \(c=1\), so the verification conditions~\eqref{eq:general-complementarity-A}--\eqref{eq:general-complementarity-B} make \((F,F)\) a scalar equilibrium. Interchangeability confines every equilibrium strategy to \(\{a,sa\}\). Proposition~\ref{prop:outcome-equivalence} fixes its mean at that of \(F\); since \(a\ne sa\), this fixes both probabilities. Thus \((F,F)\) is the unique scalar equilibrium, and Proposition~\ref{thm:existence} transfers uniqueness to the seven-battlefield contest.
\end{proof}

\section{Supplement to Section~\ref{S5:majority}}\label{app:j7-numerical-F}

\paragraph{A Numerical Approximation for Seven-Battlefield Majority.}
For \(J=7\) and \(r=1\), \(\varphi(t)=(35t^4+21t^5+7t^6+t^7)/(1+t)^7\) and \(\varphi'(t)=140t^3/(1+t)^8\).
By Proposition~\ref{prop:high-elasticity} and Theorem~\ref{thm:infinite-phase}, every \(F\in\mathcal S_{7,1}\) has zero normalized payoff and mean normalized effort \(1/2\). Writing \(F=\sum_{k\geq0}p_k\delta_{x_k}\) with \(x_k>0\), equilibrium requires
\[
V_F(x)\equiv\sum_{k\geq0}p_k\varphi(x/x_k)-x\leq0
\quad(x\geq0),
\qquad
V_F(x_k)=0,
\]
together with \(\sum_kp_k=1\) and \(\sum_kp_kx_k=1/2\). At every positive support point,
\[
V_F'(x_i)=\sum_{k\geq0}\frac{140p_kx_i^3x_k^4}{(x_i+x_k)^8}-1=0.
\]

We start from a finite-support approximation on a grid equally spaced in the logarithm of effort, repeatedly add profitable deviations and re-solve, then refine efforts and probabilities using approximate indifference and first-order conditions. The resulting eight effort levels and their probabilities approximate the unique exact equilibrium distribution \(F^*_{7,1}\) in Theorem~\ref{thm:majority-regimes}. The theorem establishes uniqueness of \(F^*_{7,1}\), but does not certify the accuracy of this approximation. Adjusting its mean to \(1/2\) gives\par
{\footnotesize
\setlength{\arraycolsep}{2pt}\[
\begin{array}{c|cc}
k & x_k & p_k \\ \hline
0 & 0.522393974289542 & 0.955202635416346 \\
1 & 0.023645476863496 & 0.042514577635406 \\
2 & 0.001209456699792 & 0.002161431193108 \\
3 & 6.43914600647\times10^{-5} & 1.14799182685\times10^{-4}
\end{array}
\hspace{1em}
\begin{array}{c|cc}
k & x_k & p_k \\ \hline
4 & 3.48104677628\times10^{-6} & 6.20007819873\times10^{-6} \\
5 & 1.89347957576\times10^{-7} & 3.37111614297\times10^{-7} \\
6 & 1.03253418805\times10^{-8} & 1.83800039742\times10^{-8} \\
7 & 5.63216511709\times10^{-10} & 1.00263912114\times10^{-9}
\end{array}
\]\par
}

The displayed distribution has total probability one and mean effort \(1/2\) at the reported precision. Evaluating \(V_F\) at numerically located positive roots of \(V_F'\) and at \(x=0,1\) gives a largest computed gain of about \(2.6\times10^{-12}\), not a rigorous upper bound on all deviation gains. For the last few reported atoms, effort ratios are about \(0.055\) and probability-to-effort ratios about \(1.78\); these finite observations do not establish limiting ratios. Proposition~\ref{prop:high-elasticity} supplies exact existence and infinite support; the finite table is only an approximation, not an exact equilibrium.

\section{Supplement to Section~\ref{S6:beyond-one}}\label{OA:beyond-one}\label{OA:general-r-outcomes}

This section allows any finite \(r>0\) and proves the existence and common-outcome assertions in Section~\ref{S6:beyond-one}.

The following proposition establishes uniform effort in expectation, not uniformity of each realized vector when \(r>1\).

\begin{proposition}[Existence and outcome equivalence for every finite discrimination parameter]\label{prop:oa-general-r-outcomes}
For every finite \(J\geq1\), admissible prize schedule, \(c\in(0,1]\), and finite \(r>0\), the unrestricted effort-vector game has a Nash equilibrium. Equilibrium strategies are interchangeable. In every equilibrium,
\begin{equation}\label{eq:oa-vector-coordinate-balance}
c\mathbb E[e_{Aj}]=\mathbb E[e_{Bj}]>0
\qquad(j=1,\ldots,J),
\end{equation}
and each contestant's expected battlefield effort is the same across battlefields and equilibria. Expected total efforts, prize shares, effort costs, and net payoffs are consequently identical across equilibria, as in Proposition~\ref{prop:outcome-equivalence}.
\end{proposition}

\begin{proof}
The equalization step in Proposition~\ref{thm:existence} is unavailable when \(r>1\), so existence must be established directly in the vector game. The constant-sum and averaging arguments do not require \(r\leq1\); we record only the changes needed for vector actions.

Write \(x_j=e_{Aj}/\Lambda\), \(y_j=e_{Bj}/\Lambda\), and let \(\mathcal U(\mathbf x,\mathbf y)\) be \(A\)'s normalized expected prize. Use the same transformation \(K_c=\mathcal U-c\sum_jx_j+\sum_jy_j\). Bounded prizes permit restriction to the cubes \([0,1/c]^J\) and \([0,1]^J\). Smooth each battlefield probability to
\[
p^\epsilon(x_j,y_j)=\frac{(x_j+\epsilon)^r}{(x_j+\epsilon)^r+(y_j+\epsilon)^r}.
\]
The compactness and minimax argument in the proof of Proposition~\ref{thm:existence} gives limiting strategies \(F_A,F_B\) and a value \(v\). For fixed own vectors with every coordinate positive, uniform convergence in the rival vector gives \(K_c(\mathbf x,F_B)\leq v\leq K_c(F_A,\mathbf y)\).

The additional boundary issue is that any subset of coordinates may vanish. Replace each zero coordinate of \(\mathbf x\) by \(\delta>0\), obtaining \(\mathbf x^\delta\). Against a positive rival coordinate its limiting win probability is unchanged; against a zero rival coordinate a certain win replaces a tie. Monotonicity and bounded convergence give
\[
K_c(\mathbf x,F_B)\leq\lim_{\delta\downarrow0}K_c(\mathbf x^\delta,F_B)\leq v.
\]
The corresponding approximation for \(B\) gives \(K_c(F_A,\mathbf y)\geq v\). Integration, domination outside the cubes, and averaging over mixed deviations complete existence exactly as before. The interchangeability proof of Lemma~\ref{lem:core-reductions}\emph{(ii)} applies unchanged.

Write \(p(x,y)=x^r/(x^r+y^r)\) for \(x+y>0\), with \(p(0,0)=1/2\). For coordinate cost balance, let \(N_j\) count \(A\)'s victories on the other battlefields and put \(D_j=\mathbb E[\widetilde q_{N_j+1}-\widetilde q_{N_j}\mid\mathbf x_{-j},\mathbf y_{-j}]\). Conditional independence makes expected prize affine in the local win probability, with \(0\leq D_j\leq1\); symmetry gives the same coefficient for \(B\). Define \(m_j=rD_jp(x_j,y_j)[1-p(x_j,y_j)]\) when \(x_jy_j>0\), and zero otherwise. Then \(0\leq m_j\leq r/4\), and the own-coordinate prize derivative is bounded by \(r/(4\eta)\) whenever own effort is at least \(\eta>0\). Thus the first-order-condition averaging in Appendix~\ref{app:proof-general-taxonomy} remains valid coordinate by coordinate, without concavity, and, for independent equilibrium draws \(\mathbf X\sim F_A\) and \(\mathbf Y\sim F_B\), gives
\[
c\mathbb E[X_j]=\mathbb E[m_j(\mathbf X,\mathbf Y)]=\mathbb E[Y_j].
\]
At zero coordinates the weighted identity is zero on both sides; no boundary first-order condition is imposed. Multiplication by \(\Lambda\) gives~\eqref{eq:oa-vector-coordinate-balance}.

Applying this identity to interchangeable equilibrium strategies fixes every coordinate mean across equilibria. Simultaneous battlefield permutations make these means identical across coordinates. They are positive: otherwise both vectors vanish almost surely, against which a small positive effort on every battlefield is profitable. Summing the identities and applying the unchanged argument of Proposition~\ref{prop:outcome-equivalence} gives all asserted common outcomes.
\end{proof}

\end{document}